\documentclass[numbib]{imaiai_hal}
\jno{}
\received{XX XX XXXX}
\revised{XX XX XXXX}
\accepted{XX XX XXXX}

 \usepackage{mathtools}
\usepackage{amsfonts} 
\usepackage{amssymb} 
  \usepackage{mathrsfs}
\usepackage[usenames,dvipsnames]{xcolor}
\usepackage[colorlinks=true,linkcolor=blue,citecolor=blue,urlcolor=blue]{hyperref}
\usepackage{hyperref}

\newcommand\bR{\mathbb{R}}
\newcommand\bN{\mathbb{N}}

\newcommand\bC{\mathbb{C}}

\newcommand\sH{\mathcal{H}}

\newcommand\sC{\mathcal{C}}

\newcommand\sS{\mathcal{S}}
\newcommand\sB{\mathcal{B}}

\newcommand\AmbientSpace{\mathcal{D}^*}
\newcommand\TestSpace{\mathcal{D}}
\newcommand\weaks{\mbox{weak-*} }

\newcommand*\id{\mathop{}\!\mathrm{d}}
\newcommand\RIPcst{\gamma}

\newcommand{\tam}{\mathrm{argmin}}

\newcommand{\ls}{\langle}
\newcommand{\rs}{\rangle}

\newcommand{\re}{\mathcal{R}e}

\newcommand{\tr}{\textrm{tr}}

\newcommand{\tos}{ \overset{*}{\rightharpoonup }}

\newtheoremstyle{assumption}{6pt}{6pt}{\rm}{}{\sffamily}{ }{ }{}
\theoremstyle{assumption}
\newtheorem{assumption}{\sc Assumption}[section]

\AtBeginDocument{
  \label{CorrectFirstPageLabel}
  \def\fpage{\pageref{CorrectFirstPageLabel}}
}
  \setcitestyle{numbers, open={[},close={]}}

\begin{document}
\title{The basins of attraction of the global minimizers of non-convex inverse problems with low-dimensional models in infinite dimension}

\author{{\sc Yann Traonmilin$^{1,*}$, Jean-François Aujol$^{1}$ and Arthur Leclaire$^{1}$ } \\
$^1$Univ. Bordeaux, Bordeaux INP, CNRS,  IMB, UMR 5251,F-33400 Talence, France.\\
$^*$ Contact author : \email{yann.traonmilin@math.u-bordeaux.fr}}
\shorttitle{Basins of attraction of the global minimizers of non-convex inverse problems} 
\shortauthorlist{Y. Traonmilin, J.-F. Aujol and A. Leclaire} 
\maketitle

\begin{abstract}
{Non-convex methods for linear inverse problems  with low-dimensional models have emerged as an alternative to convex techniques. We propose a theoretical framework where  both finite dimensional and infinite dimensional linear inverse problems can be studied. We show how the size of the  basins of attraction of the minimizers of such problems is linked with the number of available measurements. This framework recovers known results about low-rank matrix estimation and off-the-grid sparse spike estimation, and it provides new results for Gaussian mixture estimation from linear measurements. }
{ low-dimensional models, non-convex methods, low-rank matrix recovery, off-the-grid sparse recovery, Gaussian mixture model estimation}
\end{abstract}

\section{Introduction}

 Many inverse problems can be modeled as follows. From $m$ noisy linear measurements $y \in \bC^m$ defined by a projection on functions $(\alpha_l)_{1\leq l \leq m}$ (e.g. Fourier measurements): 
\begin{equation}\label{eq:measurement}
 y_l = \ls x_0, \alpha_l\rs +e_l 
\end{equation}
where $e = (e_l)_{1 \leq l \leq m}$ is an additive  noise with  finite  energy,
we aim at recovering the unknown $x_0$. This model is particularly used for imaging problems where the signal (e.g. a sound, an image, etc) must be recovered from digital measurements. The linear form $x \to \ls x,\alpha_l \rs $ typically models the response of the $l$-th sensor for a signal $x$. Let $\TestSpace$ be a space containing functions used to  measure~$x_0$ (e.g. a Banach space of smooth functions in infinite dimension or a set of vectors in finite dimension).  The measurement described by Equation~\eqref{eq:measurement} makes sense for any signal $x_0$ living in the dual space $\AmbientSpace$ of~$\TestSpace$. The bracket $\ls \cdot,\cdot \rs$ is then a duality product between $\AmbientSpace$ and $\TestSpace$.

In our framework, the space $\AmbientSpace$ is a  locally convex topological vector space with \weaks topology (we will recall in Section~\ref{sec:basin} some tools that are relevant for our study).  The measurement process is summarized  
\begin{equation}
 y= Ax_0 + e,
\end{equation}
where the linear operator $A$ is a weakly-* continuous linear measurement operator from $\AmbientSpace$ to $\bC^m$ defined, for $l=1, \ldots ,m$, by
\begin{equation}\label{eq:def_A}
(Ax_0)_l :=  \ls x_0, \alpha_l\rs.
\end{equation}
This weak topology is natural for the study of inverse problems in spaces of measures and distributions where many signals can be modeled (e.g. off-the-grid spikes~\cite{Candes_2013}). We will see in particular that the related weak notion of differentiability is sufficient to study descent algorithms that we will consider in this article.

The theory of inverse problems with low-dimensional models has shown that it is possible to recover $x_0$  when it belongs to a low-dimensional model $\Sigma$ with the procedure
\begin{equation} \label{eq:minimization}
    x^* \in \underset{x \in \Sigma}{\tam} \|Ax-y\|_2^2 
\end{equation}
provided $A$ is adequately chosen (e.g. fulfills a restricted isometry property (RIP) on the secant set $\Sigma-\Sigma$~\cite{Bourrier_2014}, see Section~\ref{sec:secant_RIP}). The estimation method~\eqref{eq:minimization} is called an ideal decoder for the considered inverse problem and low-dimensional model. It has been shown in very generic settings that it is possible to build compressive measurement operators having the required restricted isometry property  for  low-dimensional recovery~\cite{Eftekhari_2013, Puy_2017, Gribonval_2017,gribonval2020statistical}. 

In imaging applications, the goal is often to guarantee that  $x^*$ is close to $x_0$ at a given precision. To describe this, we suppose that such guarantees can be described within a Hilbert space $(\sH, \|\cdot\|_\sH)$ such that $\Sigma \subset \sH $ (the Hilbert space assumption could be dropped to a metric space setting in our proofs but all our examples fall within the Hilbert space case).  In other words, we want to ensure that the non-convex decoder~\eqref{eq:minimization} satisfies
\begin{equation}\label{eq:perf_bound}
 \|x^*-x_0\|_\sH^2\leq C \|e\|_2^2,
\end{equation}
where $C$ is an absolute constant with respect to $e$ and $x_0 \in \Sigma$.

We place ourselves in a context where the number of measurements, either  deterministic or random, guarantees that~\eqref{eq:perf_bound} is obtained with the non-convex decoder~\eqref{eq:minimization}, under a RIP assumption. The RIP is usually guaranteed by using a sufficient number of measurements with respect to the dimension of the low-dimensional model $\Sigma$. This assumption has been a cornerstone of the qualitative study of compressed sensing, sparse recovery~\cite{Foucart_2013} and general inverse problems with low-dimensional models~\cite{Traonmilin_2018}. 

Even if the decoder~\eqref{eq:minimization} is guaranteed to recover $x_0$ (up to the noise level), it is in general not convex and thus difficult to evaluate.
To cope with that, one can try to find a convex regularized minimization problem with similar recovery guarantees. Although very successful in some examples (sparse recovery in finite dimension), this approach leads  to algorithms that have computational scaling problems in some other examples (off-the-grid sparse recovery). Another general difficulty is the choice of the right convex regularization given a low-dimensional model~\cite{Traonmilin_2018a,Traonmilin_2018b}.
Another approach is to directly perform optimization~\eqref{eq:minimization} with a simple initialization followed by a descent algorithm procedure. This non-convex approach has been proposed for low-rank matrix factorization~\cite{Zhao_2015,Bhojanapalli_2016}, blind deconvolution~\cite{Ling_2017}, blind calibration \cite{Cambareri_2018}, phase recovery~\cite{Waldspurger_2018} and off-the-grid sparse spike estimation~\cite{Traonmilin_2019a,Traonmilin_2019b}. 

We propose a unified framework that follows the same idea. We consider inverse problems where the low-dimensional model can be parametrized by $\bR^d$ and we propose a general study of gradient descent in the parameter space (that can be easily extended to other descent algorithms). 

\subsection{Parametrization of the model set $\Sigma$} \label{sec:parametrization}

Let the low-dimensional model set $\Sigma \subset \AmbientSpace$ be a  cone (an assumption we make throughout the whole article) and $x_0 \in \Sigma$. Cones  are positively homogenous sets, i.e. for any $x \in \Sigma$ and $\lambda >0$, $ \lambda x \in \Sigma$.

We consider a particular (yet wide) class of inverse problems, where the low-dimensional model can be described by a  (possibly constrained) parametrization in $\bR^d$.

\begin{definition}[Parametrization of $\Sigma$]
A parametrization of  $\Sigma$ is a function ${\phi : \bR^d \to \AmbientSpace}$ such that $\Sigma \subset \phi(\bR^d) = \{\phi(\theta) : \theta \in \bR^d \}$.
\end{definition}

Our goal is to study the optimization  problem~\eqref{eq:minimization} in the parameter space.
\begin{definition}[Local minimum]
The point $\theta \in \bR^d$ is a local minimum of $g : \bR^d \to \bR$ if there is $\epsilon > 0 $ such that for any $\theta' \in \bR^d$ such that $\|\theta-\theta'\|_2 \leq \epsilon$, we have $ g(\theta) \leq g(\theta')$.
\end{definition}

 We define the reciprocal image of $\Sigma$ in the parameter space as 
\begin{equation} 
  \Theta:= \phi^{-1}(\Sigma)
\end{equation}
and the parametrized functional 

\begin{equation} \label{eq:def_g}
      g(\theta)  := \|A\phi(\theta)-y\|_2^2  .
\end{equation}

We consider the problem
\begin{equation} \label{eq:minimization2}
    \theta^* \in \arg \min_{ \theta \in \Theta}  \|A\phi(\theta)-y\|_2^2 .
\end{equation}
 As we study descent algorithms in $\bR^d$, we suppose in this article that $\Theta$ is an open set of $\bR^d$. This guarantees that the gradient of $g$ is $0$ at $\theta^*$ even when $ \Theta \varsubsetneq \bR^d$ and that $\phi(\theta^*)$ is a minimizer of~\eqref{eq:minimization}.

The model we have just described encompasses the following situations that will be studied in details within our framework in Section~\ref{sec:application}.
\begin{itemize}
 \item Low-rank symmetric positive semi-definite (PSD) matrix estimation.  We set ${\TestSpace = \AmbientSpace = \bR^{p\times p}}$,  $d = p \times r$, $\phi(Z)= ZZ^T$ ($\Theta$ is identified with the set of $p \times r$ matrices and $\Sigma = \Sigma_r$ is the set of PSD matrices of rank lower than $r$), see Section~\ref{sec:LR}.
 \item Sparse off-the-grid estimation.  The space $\TestSpace$ is the set $\sC_b^2(\bR^p)$ of twice-differentiable bounded functions on $\bR^p$ with bounded derivatives. The space $\AmbientSpace$ contains the set of compactly-supported distributions on $\bR^p$ of order less than 2. We have that $d=k(p+1)$, $\phi(a,t) = \sum_{i=1}^k a_i \delta_{t_i}$, $\Sigma= \Sigma_{k,\epsilon}$ the set of $\epsilon$-separated sums of $k$ spikes, see Section~\ref{sec:SR}.  
 \item Gaussian mixture modeling from compressed data set. We have $\TestSpace=\sC_b^2(\bR^p)$ and $\AmbientSpace$ contains the space of signed measures over $\bR^p$, $d=k(\frac{p(p+1)}{2}+p+1)$, $\phi(w,t,\Gamma) = \sum_{i=1}^k w_i \mu_{t_i,\Gamma_i}$ where $\mu_{t,\Gamma}$ is the Gaussian distribution with mean $t$ and covariances $\Gamma =(\Gamma_1,...,\Gamma_k)$. The set $\Sigma = \Sigma_{k,\epsilon,\rho,P}$ is the set of $\epsilon$-separated (with respect to an appropriate metric) sums of $k$ Gaussian distributions with eigenvalues of covariances bounded in $(\rho,P)$. See Section~\ref{sec:GMM} for the study  of Gaussian mixture models (GMM) with fixed covariance.  
\end{itemize}

In order to link the gradient of $g$ with properties of $A$, we must be able to apply a chain rule that uses the derivatives of $\phi$. The \weaks topology permits to define such derivatives in a weak sense. The important features  needed for  $\phi$ will be to follow this  weak differentiability assumption (see Section~\ref{sec:basin}), a local Lipschitz behavior around the global minimum and local boundedness properties on its derivatives. 

Note that $\phi$ is not injective in general.
A consequence is that the differential of $\phi$ might have a non trivial kernel.
This requires to adapt conventional convergence proofs to this generic setting (see Section~\ref{sec:indeterminacy}).

\subsection{Basin of attraction and descent algorithms}\label{sec:basin_def}
To perform the minimization~\eqref{eq:minimization2}, we consider the gradient descent with  fixed step $\tau$
\begin{equation} \label{def_thetan}
\theta_{n+1} = \theta_n - \tau \nabla g(\theta_n)
\end{equation}
where  $\theta_0 \in \bR^d$ is the initialization. Note that any descent algorithm could benefit from our framework (e.g. one can consider block coordinate descent to deal with the case of Gaussian mixtures with variable covariances). We choose the fixed step gradient descent for the simplicity of the analysis.

As only recovery of $x^* \in \Sigma$ matters to us, we propose the following definition of basin of attraction as we will work under conditions where any minimizer of $g$ will lead to recovery guarantees~\eqref{eq:perf_bound}.
\begin{definition}[Basin of attraction] \label{defbassin}
We say that  a set $\Lambda \subset \bR^d$ is a $g$-basin of attraction of $\theta^* \in \Lambda$  if there exists  $\tau_0>0$ such that for any $\tau \in (0,\tau_0]$, if $\theta_0 \in \Lambda$, then  the sequence $g(\theta_n)$ with $\theta_n$ defined by \eqref{def_thetan} converges to $g(\theta^*)$.
\end{definition}

This notion of basin of attraction is specific to this work in order to manage the potential indeterminacies of the parametrization. In terms of performance of the estimation, for any initialization in a $g$-basin of attraction of $\theta^*$, we will have  (Corollary~\ref{cor:recovery_guarantees})
\begin{equation} \label{eq:conv_speed1}                                                                                                                                                                                                                                                               
\|\phi(\theta_n) -x_0\|_\sH^2 \leq C\|e\|_2^2 +O\left(\frac{1}{n}\right).
\end{equation}

In other words, the gradient descent leads to an estimation of $x_0$ that verifies the recovery guarantee~\eqref{eq:perf_bound} with respect to the norm $\|\cdot\|_\sH$ that was chosen to quantify the estimation performance of minimization~\eqref{eq:minimization}, provided the initialization is in a basin of attraction of $\theta^*$.

Following classical optimization results (see e.g. \cite{Bauschke_2011, ciarlet1989introduction}), 
an open set $\Lambda$ containing $\theta^*$ is a $g$-basin of attraction of a global minimizer $\theta^*$ if
\begin{itemize}
\item $g$ is differentiable with Lipschitz gradient; 
\item $g$ is convex on $\Lambda$ ;
\item for all $n$, $\theta_n \in \Lambda$.
\end{itemize}

To deal with indeterminacy, we will show in our main theorem that the convexity property is only needed in relevant directions chosen between the current point and the closest minimizer (within equivalent parametrizations).

 Finding a good initialization is a difficult problem that was solved practically and theoretically in selected applications such as phase recovery, and only practically for some others (e.g. of the grid super-resolution in 2D). We discuss possible leads for the systematic study of this problem in Section~\ref{sec:initialization}. Also, note that we focus on \emph{convergence} of descent algorithms in this article, thus leading to the convergence speed~\eqref{eq:conv_speed1}.   The study of faster convergence  (i.e. geometric convergence)  typically requires strong convexity, and thus would require a notion of ``strong basin of attraction", which is left for future work.

\subsection{Related work}

This work unifies recent results on descent algorithms for non-convex optimization for inverse problems with low-dimensional models in both finite and infinite dimension, such as low-rank matrix recovery, phase recovery, blind deconvolution (whose common properties are highlighted in \cite{Chi_2019}) and off-the-grid sparse spike estimation~\cite{Traonmilin_2019a,Traonmilin_2019b}.  We choose the point of view of optimization in the parameter space to keep things as simple as possible from  a practical perspective. In~\cite{Barber2018gradient},  conditions on constraints in the parameter space were given in addition to a local convexity hypothesis on the studied functional to guarantee the success of projected gradient descent. We  systematically link  the measurement operator properties and the properties of the parametrization function to give explicit basins of attraction for simple gradient descent. This requires to work in infinite dimensional spaces where the low dimensional model can live.  The authors of~\cite{Blumensath_2011} give properties of measurement operators that are sufficient for the success of iterative projections (which can be seen as a projected gradient descent) and an application for infinite union of subspaces (symmetric cones in real vector spaces). While of great interest in some of our examples (such as sparse spike estimation~\cite{Traonmilin_2019b}), performing the projection step is in general not trivial, thus motivating the study of the simpler gradient descent.

Another approach is to define the descent algorithm directly on the manifold $\Sigma \subset \AmbientSpace$~\cite{Boumal_2018} or even a lifted version of~$\AmbientSpace$~\cite{Chizat_2018}. The main difficulty with this approach is to define the gradient on the manifold, since the tangent space of $\Sigma$ might not stay in a ``natural'' ambient space. For example, for the case of recovery of separated Diracs on the space of measures, the ``tangent space'' includes distributions of order 1 (that is, distributions that only involve derivatives of order $\leq 1$ of the test functions, see \cite{Hormander_1998}) which are not measures. We define a minimal framework starting from the measurement process that allows us to study the non-convex optimization method~\eqref{eq:minimization}. Such minimal structures for regularized inverse problems in  Banach spaces have been mentioned in the case  of off-the-grid spike recovery~\cite{Duval_2015} and they have been studied precisely in~\cite{Unser_2019}. As no particular metric is needed for the recovery process (only for measuring the success of recovery) on $\AmbientSpace$, we can give our result with only the \weaks topology on $\AmbientSpace$ and the norm used to quantify estimation errors  in $\Sigma$.

\subsection{Contributions}

We aim at giving a unified understanding at non-convex inverse problems with low-dimensional models frequently found in signal processing and machine learning in finite and infinite dimension. 
For low-dimensional models which can be parametrized in $\bR^d$:
\begin{itemize}
 \item We give a minimal framework where the gradient descent in the space of parameters can be described.
 \item We describe how basins of attraction of the global minimum can be studied, and how their size can be linked to the number of measurements in general under regularity conditions on the parametrization functional. This study is summarized by the general Theorem~\ref{th:control_general} and its Corollary~\ref{cor:control_general} (which is used for our examples).
 \item We describe how this framework can be applied to the examples of low-rank matrix recovery and off-the-grid spike super-resolution, and we give new results for the estimation of Gaussian mixture models (such results were never given in the GMM case to the best of our knowledge). 
 \item We present the general initialization technique by backprojection within our framework and we discuss its practical difficulties.
\end{itemize}

\section{Explicit basins of attraction of the global minimizers} \label{sec:basin}

We define precisely a framework where commonly encountered linear inverse problems can be studied. In this framework, we can study the non-convex minimization problem~\eqref{eq:minimization2}. In particular, we give conditions which guarantee that  explicit basins of attraction of the global minimizers of the function $g$ can be given. Notations used in the article are summarized in Section~\ref{sec:notations}.

\subsection{Definitions}
In our motivating examples, $\TestSpace$ is a Banach space. 

\begin{definition}
Let $x_n \in \AmbientSpace$. The sequence $(x_n)_n$  converges to $x\in \AmbientSpace$ for the \weaks topology if  for all $\alpha \in \TestSpace$

\begin{equation}
 \ls x_n, \alpha \rs \to  \ls x, \alpha \rs .
\end{equation}
In this case, we denote $x_n   \tos x$.
\end{definition}

By construction of $\AmbientSpace$ and its \weaks topology, the operator $A$ defined by~\eqref{eq:def_A} is a linear \weaks continuous operator over $\AmbientSpace$, which implies that for any $(x_n)_n$ such that $x_n  \tos x$, $Ax_n \to Ax$ (see Section~\ref{sec:weaks} for the precise definition of \weaks continuity).

In selected examples, the considered objects are generally not Fréchet differentiable.
We thus use the notion of \weaks Gateaux differentiability~\cite{Ekeland_1999}, which is based on directional derivatives.
   
 \begin{definition}[Differential, directional derivative]
  In  $\AmbientSpace$, a map $\phi: \bR^d \to \AmbientSpace$ is \weaks Gateaux differentiable at $\theta$ if there exists a linear map $L_\theta(\phi) : \bR^d \to \AmbientSpace$ such that for all $u \in \bR^d $, 
  \begin{equation}
  \frac{ \phi(\theta + h u ) - \phi(\theta)}{h} \underset{h\to 0}{\tos} L_\theta(\phi) u
  \end{equation}
  We write $\partial_u \phi (\theta) = L_\theta(\phi) u$, and $\frac{\partial \phi(\theta)}{\partial \theta_i}$ the derivative in the direction of the $i$-th canonical vector of~$\bR^d$.
 \end{definition}

 In the following sections, we shall assume that $\phi$ is twice \weaks Gateaux differentiable, i.e. $\phi$ is \weaks Gateaux differentiable and for  any $u$, $\partial_u \phi$ is \weaks Gateaux differentiable. Note that we will not suppose \weaks continuity of the derivatives in our analysis. 
 
 We summarize in Figure~\ref{fig:summary} the objects and structures used in this article.
 
 \begin{figure}[ht!]
\centering
\includegraphics[width=0.95\columnwidth]{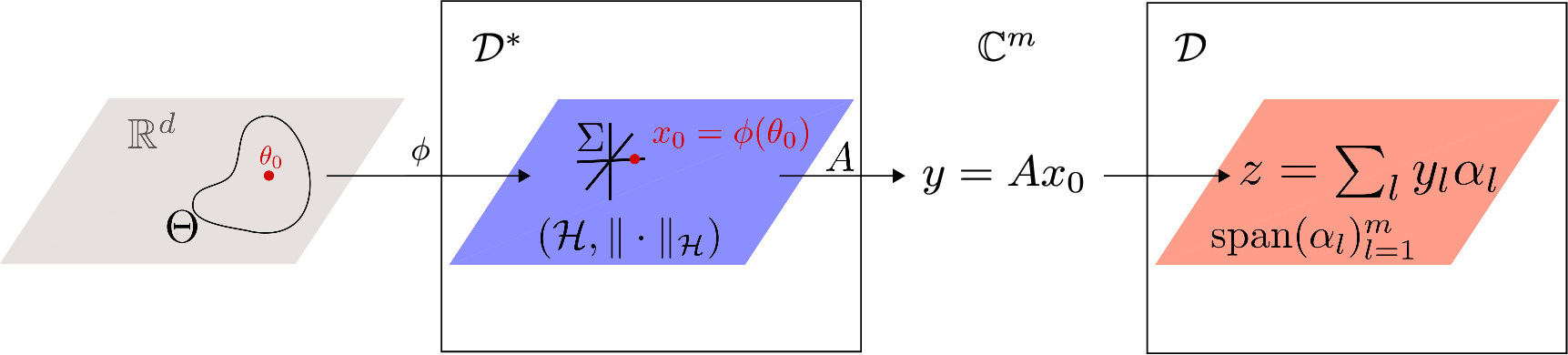}
\caption{A summary of the objects and structures of the framework. The  cone $\Sigma$ is a low-dimensional model set parametrized by $\Theta$. The measurements $y = A x_0 $ can be projected in $\TestSpace$ for initialization purpose (the ideal backprojection initialization $z$ is described in Section~\ref{sec:initialization}).
}\label{fig:summary}
\end{figure}
 
\subsection{Gradient and Hessian of the objective function }\label{sec:gradient_Hessian}

We calculate the gradient and Hessian of the function $g$ (defined in~\eqref{eq:def_g}) in the two following propositions. 
\begin{proposition}\label{prop:gradient}
  Let $A$ be a linear \weaks continuous operator from $\AmbientSpace$ to $\bC^m$ and $\phi$ a \weaks Gateaux differentiable function. Then for any $\theta \in \bR^{d}$, the function $g$ is Gateaux differentiable at $\theta$ and
\begin{equation}
  \frac{\partial g(\theta)}{\partial \theta_i }= 2 \re \ls A \frac{\partial \phi(\theta)}{\partial \theta_i } , A \phi(\theta)-y\rs.
\end{equation}
In the following we will denote $\nabla g (\theta) = \left( \frac{\partial g(\theta)}{\partial \theta_i } \right)_{1 \leq i \leq d}$ the Gateaux gradient.

\end{proposition}
\begin{proof}
See Section~\ref{sec:proof_grad}.
\end{proof}

\begin{proposition}\label{prop:Hessian}
  Let $A$ be a linear \weaks continuous operator from $\AmbientSpace$ to $\bC^m$ and $\phi$ a twice \weaks Gateaux differentiable function.
  
  For any $\theta \in \bR^{d}$, $g$ is twice  Gateaux differentiable at $\theta$ and
\begin{equation}
   H_{i,j}:=\frac{\partial^2 g(\theta)}{\partial \theta_{i} \partial \theta_{j}}
   = G_{i,j} + F_{i,j} 
\end{equation}
where
\begin{equation}
 G_{i,j}:= 2 \re \ls  A\frac{\partial \phi(\theta)}{\partial \theta_i } , A \frac{\partial\phi(\theta)}{\partial \theta_j }\rs
\end{equation}
and 
\begin{equation}
 F_{i,j}:= 2 \re \ls  A\frac{\partial^2 \phi(\theta)}{\partial \theta_i \partial \theta_j } , A \phi(\theta)-y\rs .\\
\end{equation}
\end{proposition}
\begin{proof}
See Section~\ref{sec:proof_grad}.
\end{proof}

\subsection{Secant sets and the RIP }\label{sec:secant_RIP}

The following definitions allow to express the restricted isometry property.  We then provide a fundamental lemma useful to make the connection between the RIP and the Hessian.

\begin{definition}[Secant]
  The secant set of the model set $\Sigma$ is $\sS(\Sigma) = \Sigma - \Sigma := \{ x-y : x \in \Sigma, y \in \Sigma\}$. A secant is  an element of the secant set.
\end{definition}

\begin{definition}[Generalized secant] 
Suppose $\phi$ is \weaks Gateaux differentiable. A generalized secant  is either a secant or a directional derivative $\partial_u \phi(\theta)$ with $\phi(\theta) \in \Sigma$. The generalized secant set $\overline{\sS(\Sigma)}$ is the  set of generalized secants.
\end{definition}

In the context of manifolds, the generalized secant set is linked with the tangent space of the manifold $\Sigma$: it contains the directional derivatives of elements of $\Sigma$ (with respect to their parametrization) because a tangent vector is a limit of secants.
We suppose the existence of a Hilbert space $(\sH, \|\cdot\|_\sH)$ such that $\Sigma \subset \sH$. The following assumption will be needed.
\begin{assumption}[Compatibility of $\|\cdot\|_\sH$ with generalized secants]\label{assum:k_norm}
  For all   $x \in\overline{\sS(\Sigma)}$ such that  $x= \partial_u\phi(\theta)$, we have that $\left\|\frac{\phi(\theta + |h_n| u)-\phi(\theta )}{|h_n|}\right\|_\sH $ converges for any $|h_n| \to 0$ to a limit that  does not depend on the choice of the real sequence $h_n$. This limit is written $\| \partial_u\phi(\theta)\|_\sH$.
\end{assumption}
With this assumption, we can extend $\|\cdot\|_\sH$ to  $\overline{\sS(\Sigma)}$.

While trivial in finite dimension, the last assumption must be considered carefully in infinite dimension (see Section~\ref{sec:application}). We now have sufficient tools to define the RIP.

\begin{definition}[RIP]
The operator $A$ has the RIP on $\sS(\Sigma)$ with respect to $\|\cdot\|_\sH$  with constant $\RIPcst$ if for all $x \in \sS(\Sigma)$
\begin{equation}\label{eq:DefRIP}
(1-\RIPcst)\|x\|_\sH^2 \leq \|A x\|_2^2 \leq (1+\RIPcst)\|x\|_\sH^2.
\end{equation}
\end{definition}

The RIP is very useful for the qualitative study of inverse imaging problems:  measurement operators $A$ are chosen such that the RIP constant $\RIPcst$ improves (i.e. decreases) when the number of measurements increases. In  many compressed sensing examples,  it can be guaranteed that appropriately chosen random operators $A$ have the RIP with high probability as long as $m \geq O(d\text{polylog}(d))$, i.e. the number of measurements is of the order of the intrinsic dimension $d$ of $\Sigma$, with a  dependency  on the dimension of $\sH$ (or the choice of $\|\cdot\|_\sH$ in infinite dimension) in $log$ factors.
Thanks to the compatibility assumption~\ref{assum:k_norm}, we can extend  the RIP  to $\overline{\sS(\Sigma)}$ which contains the directional derivatives of $\phi$.

\begin{lemma}[RIP on the generalized secant set]\label{lem:RIP_derivative}
 Suppose $A$  has the RIP on  $\Sigma-\Sigma$ with constant $\RIPcst$ and $\phi$ is \weaks Gateaux differentiable.  Suppose $A$ is \weaks continuous. Suppose that $\|\cdot\|_\sH$ verifies compatilbiltiy assumption~\ref{assum:k_norm}. Let $\nu \in \overline{\sS(\Sigma)}$ then
\begin{equation}
 (1-\RIPcst)\left\| \nu \right\|_\sH^2 \leq \left\|A\nu\right\|_2^2   \leq (1+\RIPcst)\left\|\nu \right\|_\sH^2.  \\
 \end{equation}
\end{lemma}
\begin{proof}
See Section~\ref{sec:proof_secant}.
\end{proof}

\subsection{Indeterminacy of the parametrization} \label{sec:indeterminacy}

The parametrization function $\phi$ is not injective in general, leading to an indeterminacy in the parametrization. Theoretical complications appear especially when the set of equivalent parameters $ \left\{\tilde{\theta } : \phi(\tilde{\theta }) = \phi(\theta)\right\}$ is not a set of isolated points, e.g. in the low-rank matrix recovery case when the factors can only be recovered up to a multiplication by an orthogonal matrix. While a basin of attraction exists, the function $g$ might even not be locally convex~\cite{Chi_2019} (e.g. in the low-rank matrix recovery case).

To cope with this indeterminacy, we  study the Hessian of $g$ in the directions $u$ relevant to the proof of convergence of the gradient descent. To do this, we introduce the following notations.  Let $\theta^*$ be a global minimizer of $g$ on $\Theta$. We define
\begin{equation}
  \label{eq:distance_phi}
  d(\theta, \theta^*) := \min_{\substack{\tilde{\theta} \in \Theta \\ \phi(\tilde{\theta}) = \phi(\theta^*) }} \| \tilde{\theta}- \theta  \|_2 \ , \quad \text{and} \quad
  p(\theta, \theta^*) := \underset{\substack{\tilde{\theta} \in \Theta \\ \phi(\tilde{\theta}) = \phi(\theta^*) }}{\tam} \|\tilde{\theta} - \theta\|_2 \subset \Theta .
\end{equation}
We will study basins of attraction having the shape 
\begin{equation}
  \label{eq:Lambda}
  \begin{split}
  \Lambda_{\beta} &:= \{ \theta \in \Theta \ : \ d(\theta, \theta^*) < \beta \}.
  \end{split}
\end{equation}

Notice that $\phi^{-1}(\{\phi(\theta^*)\})$ is a closed subset of $\Theta$ when $\phi$ is \weaks continuous, which allows to define $d(\theta, \theta^*)$ as a minimum.
Actually, $d(\theta, \theta^*)$ is the distance to the closed set $\phi^{-1}(\{\phi(\theta^*)\})$ and $p(\theta, \theta^*)$ is the (set-valued) projection of $\theta$ on $\phi^{-1}(\{\phi(\theta^*)\})$.
However, one should be warned that $d(\cdot,\cdot)$ is not a true distance function. In practice, in our examples, the set $  p(\theta, \theta^*)$ is composed of a unique element in the basin of attraction.  

This notion of distance plays an important role in the proof of our main result to show the stability of iterates of the gradient descent. The use of a $\ell^2$-based distance between parameters permits to use classical arguments relying on the scalar product.

\subsection{Control of the Hessian with the restricted isometry property}\label{sec:control_Hessian}

We begin by giving a control on the Hessian of $g$ around a minimizer $\theta^*$.

\begin{lemma}\label{lem:control_Hessian} 
Suppose $A$ is \weaks continuous and has the RIP on $\Sigma-\Sigma$. Suppose $\phi$ is twice \weaks Gateaux differentiable. Suppose that $\|\cdot\|_\sH$ verifies the compatibility assumption~\ref{assum:k_norm}. Let $\theta^*$ be a global minimizer of $g$ on the open set $\Theta$.  Let $\Lambda \subset \bR^d$ be a set such that  for all $\theta \in \Lambda$, we have $\phi(\theta)-\phi(\theta^*) \in \sS(\Sigma)$.   Let $\theta \in \Lambda$ and $H$ be the Hessian of $g$ at $\theta$. 

For  all $u\in \bR^d$, we have 
\begin{equation}
     u^THu \geq 2(1-\RIPcst) \|  \partial_u \phi(\theta)\|_\sH^2 -  2\|  A\partial_u^2 \phi(\theta)\|_2  (\sqrt{1+\RIPcst} \|\phi(\theta) -\phi(\theta^*)\|_\sH +\|e\|_2 ) \\
\end{equation}
\begin{equation}
     u^THu \leq 2 \|  A\partial_u \phi(\theta)\|_2^2 +  2\|  A\partial_u^2 \phi(\theta)\|_2  (\sqrt{1+\RIPcst} \|\phi(\theta) -\phi(\theta^*)\|_\sH +\|e\|_2 ) \\
\end{equation}
\end{lemma}
\begin{proof}
See Section \ref{sec:proof_control}.
\end{proof}

It is possible to control the Hessian of $g$ on a set $\Lambda$ with Lemma~\ref{lem:control_Hessian} in the directions which are relevant to guarantee convergence. For $\theta_1,\theta_2 \in \bR^d$, let us define the line segment $[\theta_1,\theta_2] = \{ t \theta_1 +(1-t)\theta_2: t \in [0,1] \}$. We gather technical hypotheses in the following assumption.

\begin{assumption}[Technical assumption on $\phi$ and radius $\beta$] \label{assum:technical}
Given $\theta^* \in \Theta$, $\beta >0$, we say that the technical assumption  on $\phi$ and on radius $\beta$ are fulfilled with constants $C_{\phi,\theta^*}$, $M_1$ and $M_2$ if 
\begin{enumerate}
 \item $\theta \in \Lambda_{2\beta} $ implies $\phi(\theta) \in \Sigma$ (local stability of the model set);  
 \item there is $C_{\phi,\theta^*}>0$ such that
   \begin{equation}\label{eq:local_lipschitz}
     \forall \theta \in \Lambda_{2\beta}, \quad \|\phi(\theta) -\phi(\theta^*) \|_\sH \leq C_{\phi,\theta^*}  d(\theta, \theta^*) \;\quad \text{(local control of} \|\cdot\|_\sH \text{ ) } ;
   \end{equation} 
   \item the first-order derivatives of $A \phi$ are uniformly bounded on $\phi^{-1}(\phi(\theta^*))$:
 \begin{equation}\label{eq:bound_fo_deriv}
M_1 := \sup_{\theta \in \phi^{-1}(\phi(\theta^*))}\sup_{u : \|u\|_2=1}\| A\partial_u\phi(\theta)\|_2  <+  \infty;
   \end{equation} 
 \item the second-order derivatives of $A \phi$ are uniformly bounded on $\Lambda_{2\beta}$:
 \begin{equation}\label{eq:hessian_Aphi}
M_2 := \sup_{\theta \in \Lambda_{2\beta}}\sup_{u,v : \|u\|_2=1,\|v\|_2=1}\| A\partial_v\partial_u\phi(\theta)\|_2  <+  \infty.
   \end{equation} 
   \end{enumerate}
\end{assumption}

We propose the following generic theorem to show that a set $\Lambda_\beta$ (defined by~\eqref{eq:Lambda}) is a $g$-basin of attraction.

\begin{theorem}\label{th:control_general} Consider the following two  set of hypotheses.

 \noindent \textbf{Framework hypotheses:} Let $A$ be  a \weaks continuous linear map from $\AmbientSpace$ to $\bC^m$.
  Suppose $A$ has the RIP on $\sS(\Sigma)$ with constant $\RIPcst$ and $\phi$ is \weaks continuous and twice \weaks Gateaux differentiable. Suppose that $\|\cdot\|_\sH$ verifies the compatibility assumption~\ref{assum:k_norm}.
  Let $\theta^*$ be a global minimizer of $g$ on the open set $\Theta$.
  
\noindent  \textbf{Specific hypotheses:} Assume that there exists $\beta > 0$ such that
\begin{enumerate}
 \item  the technical assumption~\ref{assum:technical}  on $\phi$ and  on radius $\beta$ is fulfilled with constants $C_{\phi,\theta^*}$, $M_1$ and $M_2$;
 \item for any $\theta \in \Lambda_{\beta}$, there exists $\tilde{\theta} \in p(\theta, \theta^*)$ such that
   \begin{equation} \label{eq:control_ratio}
     \forall z \in [\theta, \tilde{\theta}], \quad
     \frac{ (1-\RIPcst) \|  \partial_{\tilde{\theta}-\theta} \phi(z)\|_\sH^2 }{\sqrt{1+\RIPcst}\| A\partial_{\tilde{\theta}-\theta}^2 \phi(z)\|_2  }
     \geq C_{\phi,\theta^*} \beta + \frac{1}{\sqrt{1+\RIPcst}}\|e\|_2 .
   \end{equation}
\end{enumerate}
Then $\Lambda_{\beta}$ is a $g$-basin of attraction of $\theta^*$.
\end{theorem}
\begin{proof}
See Section \ref{sec:proof_control}.
\end{proof}

This theorem highlights the regularity properties and the control on the derivatives of $\phi$ that we require to ensure convergence.
However, it does not give an explicit expression of the radius of the basin of attraction at first sight.
We propose a corollary that makes it more explicit in the case when the projection $p(\theta,\theta^*)$ is composed of a  unique element when $\theta$ is in the basin, which is the case in all the examples covered in the next section.

\begin{corollary}\label{cor:control_general}
 Under the framework hypotheses of Theorem~\ref{th:control_general}, let $\beta_1 > 0$ such that
\begin{enumerate}
\item for any $\theta \in \Lambda_{2\beta_1} $, \textbf{there exists a unique $\tilde{\theta} \in p(\theta,\theta^*)$};
 \item the technical assumption~\ref{assum:technical}  on $\phi$ and on radius $\beta_1$ is fulfilled with constants $C_{\phi,\theta^*}$, $M_1$ and $M_2$; 
 \item  we have 
   \begin{equation} \label{eq:control_ratio2}
   \beta_2 := \frac{ (1-\RIPcst)}{C_{\phi,\theta^*}\sqrt{1+\RIPcst}}\inf_{\theta \in \Lambda_{\beta_1}} \inf_{z \in [\theta, \tilde{\theta}]}\left(
     \frac{ \|  \partial_{\tilde{\theta}-\theta} \phi(z)\|_\sH^2 }{\| A\partial_{\tilde{\theta}-\theta}^2 \phi(z)\|_2  }     \right) - \frac{1}{C_{\phi,\theta^*} \sqrt{1+\RIPcst}}\|e\|_2 
 >0.
   \end{equation}
Then $\Lambda_{\min (\beta_1,\beta_2)}$ is a $g$-basin of attraction of $\theta^*$.
\end{enumerate}
\end{corollary}

 \begin{proof}
See Section \ref{sec:proof_control}.
\end{proof}

\begin{remark}
This technique for the study of basins of attraction yields results when we can guarantee
\begin{equation}
\inf_{\theta \in \Lambda_{\beta_1}} \inf_{z \in [\theta, \tilde{\theta}]}\left(
  \frac{ \|  \partial_{\tilde{\theta}-\theta} \phi(z)\|_\sH^2 }{\| A\partial_{\tilde{\theta}-\theta}^2 \phi(z)\|_2  }     \right) >0
\quad \text{where} \ \tilde{\theta} \in p(\theta, \theta^*) \ \text{is unique} .
 \end{equation}
We will see that we can verify this in all our  examples. In the low-rank recovery example where indeterminacy causes problems, we control the second order derivatives in the relevant directions $u =\tilde{\theta}-\theta$. For Dirac and Gaussian estimation, we can bound uniformly the Hessian in all directions within the basin of attraction. 
\end{remark}
\begin{remark}
The fact that regularity assumptions  are on $\Lambda_{2 \beta_1}$ instead of $\Lambda_{\beta_1}$ is essentially a technical argument to guarantee the stability of the iterates in $\Lambda_{\beta_1}$  in a general theorem. It could be reduced to an assumption on $\Lambda_{\beta_1 +\eta}$ with $\eta$ small,  by reducing the step size. In our examples (Dirac and Gaussian estimation), it could be reduced to~$\Lambda_{\beta_1}$ by using a specific convergence proof. 
\end{remark}

 Note that in the noisy case, a small noise assumption (which is linked with the smallest amplitudes in $\phi(\theta^*)$ in practice)  guarantees the non-negativity of the Hessian with Lemma~\ref{lem:control_Hessian}. In the next section we will present the results in the noiseless case for clarity purpose. 

From the expression of $\beta_2$ (i.e. the size of the basin), we observe a general behavior that was already observed in the case of low-rank matrix recovery and spike estimation: when the RIP constant decreases (i.e. the number of measurements increases), the size of the basin increases (possibly not strictly). 

 In Hilbert spaces, when $\sS(\Sigma)$ has finite dimension $d$ (for the upper-box counting dimension),  it was shown in~\cite{Puy_2017} that it is possible to construct a random linear operator that have, with high probability, the  RIP on $\sS(\Sigma)$ with constant $\gamma$ such that
\begin{equation}
 m = O\left(\frac{d}{\gamma^2}\right).
\end{equation}

With such operators, we can write $\gamma =  \frac{1}{c \sqrt{m}}$ where $c$ is a constant independent of~$m$. This  gives the explicit dependency on the number of measurements. 
Note that the constants involved are dependent on the model $\Sigma$ and typically include \text{log} factors (see Section~\ref{sec:application}). 

In all our examples,  random operators following a similar relation can be constructed (even for sparse spike reconstruction and GMM estimation \cite{Gribonval_2017}).

Finally, we provide the following Corollary to show that the gradient descent leads to a solution that has the right estimation guarantees. 
\begin{corollary}\label{cor:recovery_guarantees}
 Under the framework hypotheses  and specific hypotheses of Theorem~\ref{th:control_general}, let $\theta_n$ be the iterates  (constructed in the proof of Theorem~\ref{th:control_general}) of a gradient descent with sufficiently small fixed step size such that $g(\theta_n) \to g(\theta^*)$. Then 
 \begin{equation}
   \|\phi(\theta_n) -x_0\|_\sH^2 \leq \frac{4}{1-\gamma}\|e\|_2^2 +O\left(\frac{1}{n}\right).
\end{equation}
\end{corollary}
\begin{proof}
 See Section~\ref{sec:proof_control}.
\end{proof}

 \section{Application of the framework}\label{sec:application}
 In this section, we apply our framework to three examples. We highlight how it relates to existing results and how it permits to give new ones. While suboptimal for the study of non-convex algorithms in general (e.g. for low-rank matrix recovery), we are able to give new results for non-convex recovery of low-dimensional models where such a study did not exist yet (Gaussian mixture models). For the sake of simplicity we study the three following examples in the noiseless case  $\|e\|_2 =0$.
 
\subsection{Low-rank matrix recovery} \label{sec:LR}
 The low-rank PSD matrix recovery problem falls into our analysis. We are able to give  explicit basins of attraction of global minimizers. Stronger results involving the study of all critical points of the functional  show that the global minimum is the only local minimum, thus justifying the use of stochastic descents which escape saddle points for global convergence results (see~\cite{Chi_2019} for a complete overview). We still apply our framework to this case to check its validity against well known results. It is also a first step for the understanding of more complex models such as Gaussian mixture models with low-rank modeling of covariances (see Section \ref{sec:GMM}). 
 
 The set-up is as follows: 
\begin{itemize}
\item We measure matrices with projections on $\bR^m$, i.e. $ \TestSpace =  \AmbientSpace = \bR^{p \times p}$, the duality product is the scalar product associated to the Frobenius norm, which also defines the Hilbert structure. In the context of matrices the Euclidean norm is referred to as the Frobenius norm $\|\cdot\|_\sH^2= \|\cdot\|_F^2$. The associated scalar product is $\ls X,Y\rs_F = \tr (X^TY)$.
 \item The set of rank at most $r$ PSD matrices is  $\Sigma = \Sigma_r = \{ ZZ^T : Z \in \bR^{p \times r} \}$.
 \item We can simply parametrize the model set $\Sigma$ with $\phi(\theta) = \phi(Z) = ZZ^T$ where  we identified the $d$-dimensional parameter $\theta \in \bR^{d}$ ($d = pr$) as a  matrix $Z \in \bR^{p \times r}$ . Note that given $Z_1,Z_2$ such that $\phi(Z_1)=Z_1Z_1^T = Z_2Z_2^T = \phi(Z_2)$, then there exists an orthogonal matrix $H$ such that $Z_2H =Z_1$.
 \item Low-rank matrix recovery results (e.g.~\cite{Candes_2011LR}) show that an appropriately chosen random $A$ has a RIP on $\Sigma_{2r}$ with constant $\gamma$ with high probability provided $m \geq O(\frac{pr}{\gamma^2})$, hence  $\phi(\theta^*) = x_0$ (we set  $x_0 = Z_0Z_0^T$). 
\end{itemize}
For low-rank  matrix recovery, the parametrized minimization~\eqref{eq:minimization2} is written 
\begin{equation}\label{eq:min_LR}
 \min_{Z \in \bR^{p \times r}} \|A(ZZ^T) -y\|_2^2.
\end{equation}
This minimization is often called the Burer-Monteiro method~\cite{Burer_2005}.

To apply our framework, we calculate $\partial_U \phi(Z) $ the directional derivative in the direction $U \in \bR^{p \times r}$. We have (see Section~\ref{sec:proof_LR}):
\begin{equation}
\begin{split}
 \partial_{U} \phi(Z) &= UZ^T + ZU^T \\ 
   \partial^2_{U} \phi(Z)&= 2UU^T. \\ 
\end{split}
\end{equation}

We see that  the directional derivative $\partial_{U} \phi(Z) $ is $0$ for (non-trivial) $U$ such that  $UZ^T + ZU^T =0$. In the study of non-convex low-rank matrix recovery, most complications arise from this indeterminacy of the parametrization. To  give a basin of attraction of $\theta^* = Z_0$, it is shown in the literature that the interesting directions of the Hessian are the solutions of an orthogonal Procrustes problem, i.e. it is sufficient to lower bound the Hessian in all directions $U \in \bR^{p \times r}$ that can be written $U = Z'H_0 - Z_0$ with an arbitrary $Z'$ and with $H_0= \arg \min_{H \in \mathcal{O}(r)} \|Z'H -Z_0\|_F^2$ to guarantee the success of the gradient descent~\cite{Chi_2019} ($\mathcal{O}(r)$ is the orthogonal matrix group). Qualitatively, in the parameter space, the important directions  of the Hessian are those between the unkown $Z_0$ and matrices $Z$ minimizing the Frobenius distance to $Z_0$ within its class of equivalent parametrizations $\{ \tilde{Z}: \tilde{Z}\tilde{Z}^T = ZZ^T\}$.

This idea can be used within our general framework to recover a $g$-basin of attraction of $Z_0$. In our case we will instead study the Hessian in the direction $U =  Z_0H_0-Z$  where $H_0= \arg \min_{H \in \mathcal{O}(r)} \|Z_0H -Z\|_F^2$ for $Z \in \Lambda_\beta$.  While we cannot expect the function $g$ to be convex in a neighborhood of $Z_0$~\cite{Chi_2019},   there is an underlying convexity property in these particular directions which allows to use Corollary~\ref{cor:control_general}. Within our framework, we have $d(Z,Z_0)= \|Z_0H_0 -Z\|_F^2$ and $p(Z,Z_0) =Z_0H_0$.

\begin{theorem}\label{th:basin_LR}
 Let $A$  be a linear map on $\bR^{p\times p}$ with the RIP on $\Sigma_{2r}$ with constant $\RIPcst$. Suppose $e=0$ and let $ Z_0$ be a rank-$r$ global minimizer of~\eqref{eq:min_LR}.  Let $\beta_{LR} := \frac{1}{8} \frac{  (1-\RIPcst)  (\sigma_{min}(Z_0))^2 }{(1+\RIPcst)\sigma_{max}(Z_0)}$.  Then $\Lambda_{\beta_{LR}}:= \{Z : \inf_{H \in \mathcal{O}(r)} \|ZH -Z_0\|_F <\beta_{LR} \}$ is a $g$-basin of attraction of $Z_0$.
\end{theorem}

\begin{proof}
See Section~\ref{sec:proof_LR}.
\end{proof}

Our theorem gives a result that is similar to the one of~\cite{Tu_2016}: the size of our basin of attraction depends on the smallest singular value of $\phi(\theta^*)$. It also shows explicitly the dependency on the number of measurements through the RIP constant. Our result has an added dependency on the conditioning of $Z_0$. In return, we require the RIP only on $\Sigma_{2r}$ instead of $\Sigma_{6r}$.

\begin{remark}
 The descent algorithm for low-rank matrix recovery  coined as Procrustes flow in the literature uses the expression of the gradient for measurements $\alpha_l$ obtained from symmetric measurement matrices $A$. In this case  
 \begin{equation}
 \begin{split}
  \partial_U g (Z) &= 2 \re \ls A (UZ^T +ZU^T), A \phi(Z)-y\rs= 4   \re \ls A UZ^T , A \phi(Z)-y\rs\\
 \end{split}
\end{equation}
Our analysis uses the true value of the gradient for any measurement operator $A$.
\end{remark}

\subsection{Off-the-grid sparse spike recovery}\label{sec:SR}

{Off-the-grid sparse spike recovery is at the core of imaging problems in signal processing~\cite{Candes_2013, Duval_2015, Castro_2015}. They can also be used to perform some machine learning tasks such as compressive clustering~\cite{Keriven_2017}. The size of basins of attraction is directly linked with the number of measurements through RIP constants~\cite{Traonmilin_2019a}. The proof of the results of~\cite{Traonmilin_2019a} is exactly the proof of Lemma~\ref{lem:control_Hessian} coupled with controls of the chosen norm $\|\cdot\|_\sH$ and the explicit computation of the gradient and Hessian. We recall here the set-up leading to explicit basins of attraction for this specific case. 
\begin{itemize}
\item The off-the-grid sparse signals supported on $\bR^p$ are measured by projections on twice differentiable functions with bounded derivatives  (weighted Fourier measurements) $\alpha_l \in \sC_b^2(\bR^p)=\TestSpace$. Hence $\AmbientSpace$ contains the set of compactly-supported distributions of order $\leq 2$ on $\bR^p$. 
 \item The low-dimensional model is the subset of finite signed measures over $\bR^p$ defined by $\Sigma=\Sigma_{k,\epsilon}:= \{  \sum_{i=1}^k a_i \delta_{t_i} : \|t_i-t_j\|_2 >\epsilon, t_i \in \sB_2(R) \} $, where $\sB_2(R) := \{t \in\bR^p : \|t\|_2 < R \}$. 
 
 \item The parametrization function is defined for $\theta = (a_1,\ldots,a_k,t_1,\ldots,t_k) \in \bR^{k(p+1)}$ by $\phi(\theta) = \sum_{i=1}^k a_i \delta_{t_i}$ ($a =(a_1,\ldots,a_k)$ is the vector of amplitudes, $t= (t_1,\ldots,t_k)$ defines the $k$ positions in $\bR^p$). Note that any parametrization is equivalent up to a permutation of the positions and amplitudes. 
 \item In this case,  the minimization~\eqref{eq:minimization2} must be performed on the constrained set
   \begin{equation}
     \Theta_{k,\epsilon} = \big\{  (a_1,\ldots,a_k,t_1,\ldots,t_k)  :  a_i \in \bR ,  t_i \in \sB_2(R), \|t_i-t_j\|_2 >\epsilon\big\} 
   \end{equation}
   The fact that we place ourselves within the hypotheses of Theorem~\ref{th:control_general} permits to guarantee that the gradient descent iterates stay in the constraint.
 \item The Hilbert norm for $\Sigma$ can be induced by a kernel metric defined on the space of finite signed measures over $\bR^d$ and  extended to distributions of order 2 of interest. Such a kernel metric takes the following form on a linear combination of Dirac masses
 \begin{equation}
 \begin{split}
 \left\|\sum_i {a_i}\delta_{t_i}\right\|_\sH^2 = \left\|\sum_i {a_i}\delta_{t_i}\right\|_K^2&:= \int K(t,s) \left(\sum_i {a_i}\id \delta_{t_i} ( t) \right) \left(\sum_j {a_j}\id \delta_{t_j}(s) \right)  \\
 &= \sum_i a_i a_j  K(t_i-t_j)  \\
 \end{split}
 \end{equation}
where $K(t,s) \propto e^{-\frac{\|t-s\|_2^2}{2\sigma^2}}$ is a Gaussian kernel with a  variance $\sigma^2$ that defines the precision at which we measure distances between elements of~$\Sigma$.  It was shown in~\cite{Traonmilin_2019a} that this kernel verifies the compatibility assumption~\ref{assum:k_norm}. 
\item Either random or regular Fourier measurements $A$ over $\bR^d$ can be considered. They have been shown to have a RIP on $\sS(\Sigma_{k,\epsilon})$ with constant $\gamma$ with respect to the kernel metric $\|\cdot\|_\sH$ as long as $m = O(\frac{k^2d}{\gamma^2}(\log(k))^2 \log(kd/\epsilon) )$ (random Gaussian Fourier measurements) or $m= O(\frac{1}{\epsilon^d})$ (regular Fourier measurements). 
 \end{itemize}

For off-the-grid  sparse spike estimation, the parametrized minimization~\eqref{eq:minimization2} is written 
\begin{equation}\label{eq:min_spikes}
 \min_{a_1,..,a_k \in \bR; t_1,..,t_k \in \sB_2(R); \forall i\neq j, \|t_i-t_j\|_2>\epsilon} \left\|A\left(\sum_{i=1}^k a_i \delta_{t_i} \right) -y\right\|_2^2.
\end{equation}

 The derivatives of $\phi$ are given by  $\partial_u \phi(a,t) = \sum_i v_i \delta_{t_i} + a_i \partial_{w_i}\delta_{t_i}$ where $u=(v,w)$ ($v$ is the direction for the derivative with respect to amplitudes and $w$ is the direction for the derivative with respect to positions) and $\partial_{w_i}\delta_{t_i}$ is a directional derivative of the Dirac in the distribution sense: for  $\alpha \in \TestSpace$, $\ls \partial_{w_i}\delta_{t_i},\alpha\rs = -\partial_{w_i} \alpha(t_i) $.
 
Within this framework, the case of Dirac recovery is the one with the most complications as elements of the generalized secant set $\overline{\sS(\Sigma)}$ are not elements of the set of finite signed measures which is naturally considered for off-the-grid Dirac recovery. In consequence, the considered kernel metric $\|\cdot\|_\sH$ which is defined on the space of finite signed measures, must be extended to elements of $\overline{\sS(\Sigma)}$ which are distributions of order 1 (the Dirac derivatives). Such an extension uses the smoothness of the Gaussian kernel and the fact that $\overline{\sS(\Sigma)}$ acts on a bounded domain.

We obtain the following theorem in the noiseless case.
 
\begin{theorem}\label{th:basin_spikes}
  Suppose $A$ has RIP with constant $\RIPcst$ on $\sS(\Sigma_{k,\frac{\epsilon}{2}})$.

  Let $\theta^* = (a_1,\ldots,a_k,t_1,\ldots,t_k)\in \Theta_{k,\epsilon}$ be a result of constrained minimization~\eqref{eq:min_spikes} such that $0<|a_1|\leq|a_2|...\leq |a_k|$. Then there is an explicit $\beta_{spikes}$ depending on the $a_i, \RIPcst, K, A$ such that $\Lambda_{\beta_{spikes}} :=  \{ \theta:  \|\theta-\theta^*\|_2 < \beta_{spikes} \}$ is a $g$-basin of attraction of~$\theta^*$.
\end{theorem}

We refer to~\cite[Corollary 3.1]{Traonmilin_2019a} for the proof and precise value of $\beta_{spikes}$. The value of $\beta_{spikes}$ exhibits the behavior with respect to the RIP constant mentioned in Section~\ref{sec:control_Hessian}. It also shows a dependency on the smallest amplitudes in $\theta^*$. The strong RIP assumption (i.e with separation $\frac{\epsilon}{2}$ instead of  $\epsilon$) is used to guarantee that the first condition of Theorem~\ref{th:control_general} is met (the same argument is used in the next section for GMM).

\subsection{Gaussian mixture estimation from compressive measurements}\label{sec:GMM}
 
  Gaussian mixture model recovery from  linear measurements can be used to model blind deconvolution problems. Additionally, the estimation of GMM  from compressive measurements can be used to perform compressive statistical learning.  Consider a database $x_1,\ldots, x_n \in \bR^p$ and represent it with its associated empirical  probability measure $x = \frac{1}{n}\sum_{i=1}^n \delta_{x_i}$. Let $A$ be a random Fourier measurement operator. For $m<np$, we obtain a compressed version $y=Ax$ of $y$ called a sketch of the dataset (using random Fourier features). It was shown that  the parameters of  a Gaussian mixture (GMM) can be recovered by solving the minimization~\eqref{eq:minimization} with the appropriate design of $A$~\cite{Keriven_2016,Gribonval_2017,gribonval2020statistical}. Recovery guarantees have been given for the case of a fixed known covariance PSD matrix~$\Gamma$ while practical results including the estimation of diagonal covariances were obtained using a heuristic based on orthogonal matching pursuit~\cite{Keriven_2016}. In such greedy methods, an unconstrained gradient descent step is used to refine the solution within the algorithm.  We give the expression of an explicit basin of attraction in the fixed covariance case, which gives an understanding of the success of such descent algorithms and we discuss how the result can be extended to the case of variable covariance afterwards.
 
 \begin{itemize}
 \item The Gaussian  mixtures  on  $\bR^d$ are measured by projections on bounded  functions (Fourier measurements) $\alpha_l \in \sC_b(\bR^p) =\TestSpace$. Hence $\AmbientSpace$ contains the set of finite signed measures on $\bR^p$. 
 \item The low-dimensional model is a subset of finite signed measures over $\bR^p$ defined by $\Sigma=\Sigma_{k,\epsilon,\Gamma} := \{  \sum_{i=1}^k a_i  \mu_{t_i} : \|t_i-t_j\|_\Gamma >\epsilon, t_i \in \sB_2(R) \} $, where $\id \mu_{t_i}(t)= e^{-\frac{1}{2}\|t-t_i\|_\Gamma^2} \id t $, $\|u\|_\Gamma^2 = u^t\Gamma^{-1}u$ and $\Gamma$ is the fixed known covariance matrix. 
 \item The parametrization function is defined for $\theta = (a_1,\ldots,a_k,t_1,\ldots,t_k) \in \bR^{k(p+1)}$ by $\phi(\theta) = \sum_{i=1}^k a_i \mu_{t_i}$ ($a =(a_1,\ldots,a_n)$ is the vector of amplitudes, $t= (t_1,\ldots,t_n)$ defines the means). As in the spike estimation problem, any parametrization is equivalent up to a permutation of the means and amplitudes.
 \item The minimization~\eqref{eq:minimization2} is performed on the constrained set $\Theta_{k,\epsilon} \subset \bR^{(k+1)d}$. 
 \item The Hilbert norm for $\Sigma$ can be obtained from a kernel metric $\|\cdot\|_\sH := \|\cdot\|_K$ defined on the space of finite signed measures over $\bR^d$ where the kernel is a Gaussian function with covariance proportional to $\Gamma$ that defines the precision at which we measure distances between elements of $\Sigma$. The compatibility assumption~\ref{assum:k_norm} essentially comes from the fact that the directional derivatives of Gaussian measures with respect to amplitudes and means are finite signed measures (they are obtained by differentiating their $\sC^{\infty}$ densities). 
\item Adequately chosen random Fourier measurements with Gaussian frequency measurements have been shown to satisfy the RIP  on the secant set $\sS(\Sigma_{k,\epsilon,\Gamma})$ with constant $\gamma$ with respect to $\|\cdot\|_K$ as long as $\epsilon \geq O(\sqrt{d(1+\log(k))})$ and   $m = O(\frac{k^2d}{\gamma^2}\text{polylog}(k,d))$ \cite{Gribonval_2017}.
 \end{itemize}
 
For GMM estimation, the parametrized minimization~\eqref{eq:minimization2} is written
\begin{equation}\label{eq:min_GMM}
 \min_{a_1,..,a_k \in \bR; t_1,..,t_k \in \sB_2(R); \forall i\neq j, \|t_i-t_j\|_\Gamma>\epsilon} \left\|A\left(\sum_{i=1}^k a_i \mu_{t_i} \right) -y\right\|_2^2.
\end{equation}

 The derivatives of $\phi$ are given by  $\partial_u \phi(a,t) = \sum_i v_i \mu_{t_i} + a_i\partial_{w_i} \mu_{t_i}$ where $u=(v,w)$ ($v$~is the direction for the derivative with respect to amplitudes and $w$ is the direction for the derivative with respect to means). As $\mu_{t_i}$ has a smooth density with respect to the Lebesgue measure its directional derivative $\partial_{w_i} \mu_{t_i}$ in the distribution sense is also a finite signed measure with density 
 \begin{equation}\label{eq:GMM_kernel_norm_expr1}
  t \to - w_i^T\Gamma^{-1}(t-t_i)e^{-\frac{1}{2}\|t-t_i\|_\Gamma^2}.
  \end{equation} 
  This makes the case of Gaussian mixtures slightly easier to manage than the Dirac recovery case as the  kernel metric $\|\cdot\|_K$  is well defined on the space of finite signed measures.
 
 We give an explicit uniform bound of the Hessian on a neighborhood of $\theta^*$ in this case. The main missing ingredient to give the bound is the mutual coherence of the kernel metric. It  can be shown that  an  appropriately chosen  Gaussian kernel $K$ satisfies the following assumption \cite{gribonval2020statistical} (by taking a kernel with small enough variance with respect to the separation).
 
 \begin{assumption}\label{assum:gaussian_k}
 The kernel $K$ follows this assumption if
 \begin{itemize}
  \item   $\| \mu_{t_i}\|_K=1$.
  \item  There is a constant $c_K$ such that for any $k$ means $(t_i)_{i=1}^k$ that verifies for $i\neq j$, $\|t_i-t_j\|_\Gamma>\epsilon$, we have $\| \sum_i v_i \mu_{t_i} + a_i\partial_{w_i} \mu_{t_i}\|_K^2 \geq (1-c_K) \sum_i \|v_i \mu_{t_i} + a_i \partial_{w_i}\mu_{t_i}\|_K^2$.
 \end{itemize}
  
 \end{assumption}

  In the noiseless case we get the following basin of attraction. 
 
 \begin{theorem}\label{th:control_gaussian}
 Suppose the Fourier measurement operator $A$ has RIP $\RIPcst$ on $\sS(\Sigma_{k,\frac{\epsilon}{2},\Gamma})$. Suppose $K$ is a Gaussian kernel with covariance proportional to $\Gamma$ that follows Assumption~\ref{assum:gaussian_k}. Suppose $e=0$. Let $\theta^* = (a_1,\ldots,a_k,t_1,\ldots,t_k)\in \Theta_{k,\epsilon}$ be a solution of the constrained minimization problem~\eqref{eq:min_GMM} such that $0<|a_1|\leq|a_2|...\leq |a_k|$.  Let 
 \begin{equation}
  \beta_{GMM} = \min \left(  \frac{\epsilon\sqrt{\lambda_{min}(\Gamma)}}{8 },\frac{|a_1|}{2}, \frac{  (1-\RIPcst)  (1-c_K) \min(1 , d_K |a_1|^2) }{8C_{\phi,\theta^*}\sqrt{1+\RIPcst}(\sqrt{1+\gamma}\sqrt{1+c_k} \sqrt{D_K} + 2|a_k|D_{A,K}')} \right) 
 \end{equation}
 where $\lambda_{min}(\Gamma)$ is minimum eigenvalue of $\Gamma$, $c_K$, $d_K$ and $D_K$ are constants depending only on the chosen kernel $K$, $D_{A,K}$  is an explicit constant depending on $K$ and the acquisition operator $A$ and $C_{\phi,\theta^*}>0$. Additionally suppose that for all $i$, $\|t_i\|_2 \leq R - 2 \beta_{GMM}$.

 Then  $\Lambda_{\beta_{GMM}} :=   \{ \theta:  \|\theta-\theta^*\|_2 < \beta_{GMM} \}$ is a $g$-basin of attraction of $\theta^*$.

\end{theorem}

\begin{proof}
See Section~\ref{sec:proof_GMM}.
\end{proof}

With Theorem~\ref{th:control_gaussian}, we are able to control the non-negativity of the Hessian over an explicit neighborhood of $\theta^*$. Note that the assumption on the norm of the $t_i$ is technical to guarantee the stability of iterates.  It is not a very strong assumption in practice as $\beta_{GMM}$ is generally small with respect to $R$.  Also, it could be dropped with a dedicated convergence proof. The constants involved in the expression  lead to a size of the basin of attraction having the same behavior with respect to parameters of the problem: it increases with  respect to the number of measurements (RIP constant) and it decreases with respect to the minimum weight in~$x_0$. Note that with Fourier measurements proposed in \cite{gribonval2020statistical} the constant $D_{A,K}$ is independent of $m$. A dependency on the minimum amplitude in $\theta^*$ is also observed similarly to the case of Dirac estimation. 

 \paragraph{Open questions for the extension to unknown variable covariances with low-rank constraint}  

 For real world applications, the known covariance case can be too simplistic to perform compressive statistical learning with good results. In practice, the covariance matrices of Gaussian mixtures are also estimated.  Often, a low-rank approximation is made (flat-tail) to reduce the number of estimated parameters. For example, an unknown low-rank covariances model was used in  \cite{Hui_2021a,Hui_2021} to learn an image patch model from a compressed database. This model is then used to perform patch based image denoising. 
 
  A low-dimensional model with unknwon low-rank covariances can be defined as follows. The parameters are  $\theta =(a_1,\ldots,a_k, t_1,\ldots,t_k , Z_1,\ldots,Z_k)$ where the $Z_i$ are $d \times r$ matrices used to model covariances $\Gamma_i = Z_i Z_i^T + \rho I$. In this case, the model is 
 \begin{equation}
 \begin{split}
  \Sigma_{k,\epsilon,r,\rho,P} := \Bigl\{ &\sum_{i=1}^k a_i \mu_{t_i, \Gamma_i} : a_i \in \bR, \|t_i-t_j\|_2 > \epsilon, t_i \in \sB_2(R), \rho < \lambda_j(\Gamma_i) < P ,\\
  &\text{rank}(\Gamma_i-\rho I)\leq r \Bigr\}
  \end{split}
 \end{equation}
where $\mu_{t_i, \Gamma_i} $ is the Gaussian measure of mean $t_i$ and covariance $\Gamma_i$.  Recovery guarantees of the ideal decoder have not yet been given in this more general case (and it is out of the scope of this article). To do so, one would need to show the existence of a linear measurement operator with a RIP on $\sS(\Sigma_{k,\epsilon,r,\rho,P})$.

For off-the-grid sparse recovery,  a block coordinate descent with respect to amplitudes and means is easier to implement in practice. To consider the case of variable covariances, we suggest to follow this guideline and to perform a block coordinate descent step with respect to each covariance matrix. With block coordinate descent, we would just need to make sure the Hessian with respect  to the coordinates of $Z_i$ is positive over a neighborhood of $\theta^*$~\cite{Beck_2013} to calculate explicitly a basin of attraction.  
 
 The directional derivative of $a_i\mu_{t_i,Z_iZ_i^T + \rho I}$ in the direction $(v,w,W)$ has density (proof in Section~\ref{sec:proof_GMM_var}): 
 
 \begin{equation}
 \sum_i ( v_i -  a_i w^T\Gamma_i^{-1}(t-t_i ) + a_i \frac{1}{2}(t-t_i)^T \Gamma_i^{-1} (Z_iW_i^T +W_iZ_i^T)\Gamma_i^{-1}(t-t_i) )e^{-\frac{1}{2}\|t-t_i\|_{Z_iZ_i^T + \rho I}^2}.
 \end{equation}
 
 Finding a lower bound of the kernel norm of this directional derivative is quite technical.  We observe that we would need to handle the indeterminacy of the low-rank parametrization as in Section~\ref{sec:LR}. We propose the following lemma as a first step towards building a theorem for an explicit basin of attraction in the variable flat tail covariance case. The proof of this lemma shows the technicalities involved to deal with the full variable covariance case. With this lemma it is possible to show that the Hessian blocks corresponding to the covariances in the full rank diagonal case are positive over a neighborhood of the global minimum (in other words, the radius $\beta_2$ from Corollary~\ref{cor:control_general} is strictly positive) as long as the precision of the kernel is good enough.  We leave the full study of non-convex methods for the variable covariance case with low-rank approximation including establishing  RIP recovery guarantees for future work.
 
 \begin{lemma} \label{lem:bound_variable_cov}
 Let $K$ be a Gaussian kernel with covariance $\frac{1}{\lambda}I$. Suppose $\Gamma \in \bR^{p \times p}$ is a diagonal definite  positive matrix. Let $\partial_W \mu_{t,\Gamma}$ the directional derivative of $\mu_{t,\Gamma}$  with respect to $\Gamma$ in (diagonal) direction $W$ such that $\|W\|_F =1$. There exists $L >0$ such that  if $\lambda \geq L$, it then holds:
 \begin{equation}
  \|\partial_W \mu_{t,\Gamma}\|_K^2 \geq D_{\lambda,\Gamma} 
 \end{equation}
where $D_{\lambda,\Gamma}>0$ is a constant that only depends on $\lambda $ and $\Gamma$.
 \end{lemma}
\begin{proof}
See Section~\ref{sec:proof_GMM_var}.
\end{proof}



\section{From ideal to practical backprojection initialization: the challenge of non-convex low-dimensional recovery?}\label{sec:initialization}

A common approach for the initialization of non-convex low-dimensional recovery is the  initialization by backprojection techniques. Such initialization is outdated for low-rank matrix recovery as global convergence of descent algorithms has been proven~\cite{Li_2018}. However, it is still necessary for other models such as sparse spikes, Gaussian mixtures and phase recovery (which is not developed in this article). 
Finding a practical initialization technique is essentially a case by case heuristic design problem. We investigate the usual backprojection method in the noiseless case $y= A x_0$. 

In finite dimension, backprojection techniques rely on the fact that the RIP of $A$ implies $\|(A^HA -I)x_0\|_2^2 \leq \RIPcst $ (where $A^H$ is the adjoint of $A$). In our general framework, the adjoint of $A$ does not back-project in $\AmbientSpace$, but in $\TestSpace$. We formalize such backprojection within our framework.

\begin{definition}[Ideal backprojection]
Given measurements $y = A x_0 \in \bC^m$, we define the  ideal backprojection $z \in \TestSpace$ with 
\begin{equation}
 z :=  \sum_{l=1}^m  y_l \alpha_l.
 \end{equation}
 
 \end{definition}
 
We immediately have  the following Lemma (which is a direct consequence of the RIP):
\begin{lemma}\label{lem:init}
 Suppose $A$ has the RIP with constant $\RIPcst$. Let $z \in \TestSpace$ the ideal backprojection for measurements $y = A x_0$. Then
 \begin{equation}
   (1-\RIPcst) \|x_0\|_\sH^2 \leq \ls x_0, z\rs \leq (1+\RIPcst) \|x_0\|_\sH^2.
\end{equation}
\end{lemma}

\begin{proof}[Proof of Lemma \ref{lem:init}]
\begin{equation}
\begin{split}
  \ls x_0, z\rs &=  \ls x_0, \sum_{l=1}^m  y_l \alpha_l \rs\\
  &=  \ls x_0, \sum_{l=1}^m  \ls x_0,\alpha_l \rs \alpha_l \rs\\
               &=  \sum_{l=1}^m  \overline{\ls x_0, \alpha_l \rs} \ls x_0, \alpha_l  \rs  \\
               &=  \sum_{l=1}^m  | \ls x_0, \alpha_l \rs |^2  = \|Ax_0\|_2^2 .
 \end{split}
\end{equation}
Hence, with the RIP,
\begin{equation}
\begin{split}
   (1-\RIPcst) \|x_0\|_\sH^2 \leq \ls x_0, z\rs \leq (1+\RIPcst) \|x_0\|_\sH^2 .
 \end{split}
\end{equation}
  \end{proof}
  
 Lemma~\ref{lem:init} shows that the ideal backprojection preserves the energy in $x_0$ (up to a RIP constant). Hence when the number of measurement increases, it can get arbitrarily close to the global optimum (looking through the duality product). The main challenge is to extract initial parameters $\theta_{init}$ from $z$. One possibility would be to separate this task in two steps: define $\theta_{init}$ minimizing a ``distance'' between $\phi(\theta)$ and $z$ for theoretical study  and then find practical ways to minimize such distance.  Such distance could be built using the norm $\|\cdot\|_\sH$ or even directly the duality product.   

In  phase retrieval~\cite{Waldspurger_2018},  spectral initialization techniques consist in taking the leading eigenvectors of a matrix constructed as linear combination of  backprojections of individual measurements. 

In the case of spike super-resolution a heuristic based on a sampling on a grid of the ideal backprojection followed by dimension reduction was proposed to  initialize a descent algorithm to perform spike recovery for Diracs supported on a low-dimensional domain~$\bR^p$ (e.g. $p=2$ and possibly $3,4$)~\cite{Traonmilin_2019b}. Qualitative evidence that a grid step size $\epsilon_g$ small enough permits the initialization in the basin of attraction are the only recovery results for sparse spike recovery with descent methods in the parameter space. Grid based initialization techniques also requires one algorithmic step (the backprojection) whose computational complexity $\epsilon_g^{-d}$  might be a drawback  for high-dimensional $d$ (curse of dimension). A practical alternative consists in greedy heuristics to initialize descent algorithms by greedily adding spikes one by  one (general theoretical recovery guarantees for such methods are still being investigated beyond specific examples~\cite{Elvira_2019}).  

For Gaussian mixture estimation, it is possible to imagine a grid based heuristic to estimate the means followed by clustering (to backproject onto the separation constraint). However, it is an open question to determine if the curse of dimension can be avoided in the initialization in order to outperform greedy methods~\cite{Keriven_2016}.

\section{Conclusion}
 
 We have described a generic framework to perform low-dimensional recovery from linear measurements using non-convex optimization. We showed how recent examples of the literature can be studied within this framework and we gave new results for the case of Gaussian mixture modeling. 
 
 The following open questions emerge from this work: 
 \begin{itemize}
 \item  Other applications: The objective of our framework is to be used for other applications where such results do not exist. For example, our results should apply to low-rank tensor recovery, where recent works study the RIP in this context~\cite{Rauhut_2015} and study convergence of the Burer-Monteiro method~\cite{Li_2022}. It might be also possible to apply directly our framework for inverse problems using generative models with deep neural networks (e.g. in \cite{Asim_2020}, a neural network is used to learn the parametrization function $\phi$). We would then need to be able to show the existence of a RIP on such models first.
  \item RIP based guarantees, while useful qualitative tools, are often too strong to give useful quantitative guarantees. Can our framework be extended to a non-uniform recovery case to obtain more precise estimates?  
  \item Is it possible in general to  initialize within the basin of attraction, avoiding the curse of dimension (for parametrized infinite dimensional problems), and providing a full quantitative proof of convergence in the case of super-resolution and GMM estimation? 
 \end{itemize}

 \section{Annex}
 \subsection{Summary of notations}\label{sec:notations}
 \noindent Spaces  and sets:
 \begin{itemize}
  \item $\TestSpace$: Banach space of functions used to measure the unknown $x_0$.
  \item $\AmbientSpace$: ambient space where the objects of interest (e.g. the unknown $x_0$) belong, dual space of $\TestSpace$.
  \item $\Sigma$: low dimensional model set, i.e. a cone $\Sigma \subset \AmbientSpace$.
  \item $\sS(\Sigma)$: secant set of $\Sigma$ (differences of elements of $\Sigma$).
  \item $\Theta$: open subset of $\bR^d$ which parametrizes the model $\Sigma$.
  \item $\sH$: Hilbert space containing the low dimensional model. The reconstruction error is measured using the associated Hilbert norm. 
  \item $\bC^m$: complex finite dimensional vector space containing the $m$ finite measurements. 
  \item $\Lambda_\beta$:  basin of attraction for the appropriate $\beta$ (given by our theorems). 
  \item $[\theta_1,\theta_2]$: for $\theta_1,\theta_2 \in \bR^d$, denotes the line segment $ \{ t \theta_1 +(1-t)\theta_2: t \in [0,1] \}$.
 \end{itemize}
Norms, functions and operators:
 \begin{itemize}
   \item $ \ls \cdot, \cdot\rs$: depending on context: duality product between $\TestSpace$ and  $\AmbientSpace$, conventional scalar Hermitian product in finite dimension.
 \item $\|\cdot\|_\sH$: norm associated with $\sH$.
  \item $\|\cdot\|_2$: usual $\ell^2$-norm.
  \item $d(\cdot,\theta)$: distance to the set of equivalent parametrizations of $\theta$.
  \item $p(\cdot,\theta)$: projection on the set of equivalent parametrizations of $\theta$.
  \item $\phi$: parametrization function ($\phi(\Theta) = \Sigma$).
  \item $\alpha_l$: measurement functions in $\TestSpace$ for $1 \leq l \leq m$.
  \item $A$ : linear measurement operator ( $\AmbientSpace \to \bC^m$) defined by $Ax= (\ls x, \alpha_l \rs)_{1 \leq l \leq m}$.
  \item $g$: parametrized functional we want to minimize, $g(\theta) = \|A\phi(\theta) - y\|_2^2$.
 \end{itemize}

 \subsection{ Weak-* topology and continuity}\label{sec:weaks}

\begin{proposition}[Neighborhood]
 For all $x_0 \in \AmbientSpace$, a basis of neighborhoods of $x_0$ for the \weaks topology is formed with the sets 
 \begin{equation}
  V_{\epsilon,(\alpha_i)_{1 \leq i \leq n}}(x_0) := \{ x \in \AmbientSpace : \forall i =1,\ldots, n, |\ls x-x_0,\alpha_i \rs| < \epsilon \}
 \end{equation}
for all $\epsilon > 0$ and all  $\alpha_i \in \AmbientSpace$ and  $n\in \bN$.
\end{proposition}

Continuity is defined as follows.

\begin{definition}[Weak* continuity]
Let $A : \AmbientSpace \to \bR^d$. $A$ is \weaks continuous at $x_0$ if for any $\epsilon >0$, there is a neighborhood $V(x_0)$ for the \weaks topology of $x_0$ such that 

\begin{equation}
 x \in V(x_0) \implies \|f(x)-f(x_0)\|_2 < \epsilon.
\end{equation}
\end{definition}

\subsection{Proofs for Section~\ref{sec:gradient_Hessian}}\label{sec:proof_grad}

\begin{proof}[Proof of Proposition \ref{prop:gradient}]
  Let $z \in \bC^m$, and let  $z^H$ be the Hermitian conjugate of $z$. Remark that $\theta \to A\phi(\theta)$ is a Gateaux differentiable function $\bR^d \to \bC^m$. We have
\begin{equation}
\begin{split}
  \frac{\partial g(\theta)}{\partial \theta_i } &= \left(\frac{\partial [ A\phi](\theta) }{\partial \theta_i }\right)^H  A \phi(\theta) +  (A \phi(\theta))^H\frac{\partial [ A\phi](\theta) }{\partial \theta_i }  - 2 \re \ls  \frac{\partial [ A\phi](\theta)}{\partial \theta_i } , y\rs \\
   &=    2\re \left(\frac{\partial [ A\phi](\theta) }{\partial \theta_i }\right)^H  A \phi(\theta)  - 2 \re \ls  \frac{\partial [ A\phi](\theta)}{\partial \theta_i } , y\rs \\
   &=    2 \re \ls  \frac{\partial [ A\phi](\theta)}{\partial \theta_i } , A \phi(\theta)-y\rs, \\
\end{split}
\end{equation}
By linearity  and continuity of $A$ :  $  \frac{\partial[ A\phi](\theta)}{\partial \theta_i } = A \frac{\partial \phi(\theta)}{\partial \theta_i } $
 and 
 \begin{equation}
\begin{split}
  \frac{\partial g(\theta)}{\partial \theta_i } 
   &=    2 \re \ls  A\frac{\partial \phi(\theta)}{\partial \theta_i } , A \phi(\theta)-y\rs. \\
\end{split}
\end{equation}
\end{proof}

\begin{proof}[Proof of Proposition \ref{prop:Hessian}]
Use Proposition~\ref{prop:gradient} with the properties of the Hermitian product. 
\end{proof}

\subsection{Proof for Section~\ref{sec:secant_RIP}} \label{sec:proof_secant}

\begin{proof}[Proof of Lemma \ref{lem:RIP_derivative}]
  Let $\nu = \partial_u \phi(\theta) \in \overline{\sS(\Sigma)}$. Using the fact that $\Theta$ is an open set,  let  $\nu_n =\frac{\phi(\theta + |h_n| u)-\phi(\theta )}{|h_n|}$ with $|h_n| \to 0$ be a  sequence of  real numbers. Since $\Sigma$ is a cone, $\nu_n  \in \sS(\Sigma)$. Thanks to the compatibility assumption~\ref{assum:k_norm} on $\|\cdot\|_\sH$, $\|\nu_n\|_\sH$ converges to $\|\partial_u\phi(\theta)\|_\sH$. Moreover, by definition  $\nu_n \tos \partial_u\phi(\theta) \in \overline{\sS(\Sigma)}$, and by continuity of $A$, we have that $A \nu_n \to A \nu$ w.r.t to  $\|\cdot\|_2$. Using the hypothesis,
 \begin{equation}
\begin{split}
 (1-\RIPcst)\left\| \nu_n \right\|_\sH^2 \leq \left\|A\nu_n\right\|_2^2   \leq (1+\RIPcst)\left\|\nu_n \right\|_\sH^2.  \\
\end{split}
 \end{equation}
 Taking both inequalities to the limit yields the result. 
\end{proof}
 
\subsection{Proofs for Section~\ref{sec:control_Hessian}}\label{sec:proof_control}

\begin{proof}[Proof of Lemma \ref{lem:control_Hessian}]
We write $H = F+G$ as in Proposition~\ref{prop:Hessian}.

Let $u \in \bR^d$, with the linearity of $A$ and the linearity of the Gateaux differential, we have $ \sum_i  u_i \frac{\partial \phi(\theta)}{\partial \theta_i }= \partial_u \phi(\theta) $ and 
\begin{equation}\label{eq:proof_control_0}
 \begin{split}
  u^TGu = \sum_{i,j} u_i u_j G_{i,j} &=  \sum_{i,j} u_i u_j  2 \re \ls A \frac{\partial  \phi(\theta)}{\partial \theta_i } ,A \frac{\partial\phi(\theta)}{\partial \theta_j }\rs  \\
  & =  2 \re \ls A\sum_i  u_i \frac{\partial \phi(\theta)}{\partial \theta_i } , A \sum_j u_j\frac{\partial \phi(\theta)}{\partial \theta_j }\rs  \\
  &= 2 \left\| A \sum_i  u_i \frac{\partial \phi(\theta)}{\partial \theta_i } \right\|_2^2 \\
    &= 2 \left\| A \partial_u \phi(\theta) \right\|_2^2. \\
 \end{split}
\end{equation}

By definition of the generalized secant, we have $\partial_u \phi(\theta) \in  \overline{\sS(\Sigma)}$. Using the RIP of $A$ with Lemma~\ref{lem:RIP_derivative}, we get

\begin{equation}\label{eq:proof_control_1}
 \begin{split}
 2(1-\RIPcst) \|  \partial_u \phi(\theta)\|_\sH^2 \leq      u^TGu  \leq  2(1+\RIPcst) \|  \partial_u \phi(\theta)\|_\sH^2 . \\
 \end{split}
\end{equation}

Moreover, with the Cauchy-Schwarz and triangle inequalities,
\begin{equation}
 \begin{split}
  |u^TFu| & = |2\ls A\partial_u^2 \phi(\theta),A \phi(\theta)-y\rs| \\
  &\leq 2 \|  A\partial_u^2 \phi(\theta)\|_2  (\|A \phi(\theta) -A\phi(\theta^*)\|_2 + \|A \phi(\theta^*) -y\|_2 ).\\  
   \end{split}
\end{equation}
By definition of $\theta^*$, $ \|A \phi(\theta^*) -y\|_2 \leq \|A x_0 -y\|_2 =\|e\|_2$. With the RIP, 
\begin{equation}\label{eq:proof_control_2}
 \begin{split}
 |u^TFu|  &\leq 2 \|  A\partial_u^2 \phi(\theta)\|_2  (\sqrt{1+\RIPcst}\|\phi(\theta) -\phi(\theta^*)\|_\sH +\|e\|_2 ).\\
 \end{split}
\end{equation}
Combining~\eqref{eq:proof_control_1} and~\eqref{eq:proof_control_2} gives the lower bound. Combining the last equality of~\eqref{eq:proof_control_0} and~\eqref{eq:proof_control_2} gives the upper bound. 

\end{proof}

\begin{proof}[Proof of Theorem~\ref{th:control_general}]

  We start the proof by showing that given $\theta \in \Lambda_{\beta}$, there exists a global minimizer $\tilde{\theta} \in \Theta$ such that $g$ is convex on the segment $[\tilde{\theta},\theta]$. We then show the Lipschitz property of the gradient. These two facts lead to the proof of the stability of iterates in $\Lambda_{\beta}$ and permit to show the convergence of the fixed step gradient descent.
  \bigskip

\noindent\textbf{Convexity property.} Let $\theta \in \Lambda_{\beta}$ and $\tilde{\theta} \in p(\theta, \theta^*)$ (see Equation~\eqref{eq:distance_phi}) so that $u := \tilde{\theta} - \theta$ satisfies $\|u\|_2<\beta$.
By hypothesis, we can assume that $\tilde{\theta}$ satisfies~\eqref{eq:control_ratio}.
Let also $t \in [0,1]$.
We have  ${\|\theta+ tu -  \tilde{\theta}\|_2 = (1-t) \|\theta - \tilde{\theta}\|_2  < \beta}$. This implies $\theta+t u \in \Lambda_\beta$ and $\phi(\theta+t u)  \subset \Sigma$.
Denoting by $H_{\theta+t u}$ the Hessian of $g$ at $\theta+t u$, Lemma~\ref{lem:control_Hessian} and technical assumption~\ref{assum:technical} number 2 imply that
\begin{equation}
\begin{split}
 u^TH_{\theta +t u} u &\geq  2(1-\RIPcst) \|  \partial_u \phi(\theta+t u)\|_\sH^2 -  2\|  A\partial_u^2 \phi(\theta+t u)\|_2  (\sqrt{1+\RIPcst} \|\phi(\theta+t u) -\phi(\theta^*)\|_\sH +\|e\|_2 ) \\
 & \geq 2 (1-\RIPcst) \|  \partial_u \phi(\theta+t u)\|_\sH^2 -  2\|  A\partial_u^2 \phi(\theta +t u)\|_2  (\sqrt{1+\RIPcst} C_{\phi,\theta^*} d(\theta+t u , \theta^*) +\|e\|_2 ) \\
 & \geq 2 (1-\RIPcst) \|  \partial_u \phi(\theta+t u)\|_\sH^2 -  2\|  A\partial_u^2 \phi(\theta +t u)\|_2  (\sqrt{1+\RIPcst} C_{\phi,\theta^*} \|\theta+t u-\tilde{\theta}\|_2+\|e\|_2 ) \\
 & \geq  2(1-\RIPcst) \|  \partial_u \phi(\theta +t u)\|_\sH^2 -  2\|  A\partial_u^2 \phi(\theta +t u)\|_2  (\sqrt{1+\RIPcst} C_{\phi,\theta^*}\beta+\|e\|_2 ) .
\end{split}
\end{equation}
Therefore, using  Hypothesis~\eqref{eq:control_ratio}, we get that $u^TH_{\theta +t u} u \geq 0$ for all $t \in [0,1]$, which guarantees the convexity of $g$ on $[\tilde{\theta},\theta]$. 
Since $g$ is Gateaux differentiable, the properties of one-dimensional convex functions give
\begin{equation}\label{eq:control_convexity}
  g(\tilde{\theta}) -g(\theta)  \geq \ls \nabla g(\theta)  , \tilde{\theta} -\theta\rs.
\end{equation}
 While $g$ might not be convex in $\Lambda_{\beta}$ due to indeterminacies, this generalized convexity inequality  is enough to prove the convergence of the gradient descent, if we have the required Lipschitz nature of the gradient of $g$, which we address now. 
\bigskip

 \noindent\textbf{Lipschitz gradient property.}
Let $u$ be a vector such that $\|u\|_2=1$. Let $\theta \in \Lambda_{2\beta}$ and $\tilde{\theta} \in p(\theta, \theta^*)$ such that $\|\theta-\tilde{\theta}\|<2\beta$.

With Lemma~\ref{lem:control_Hessian} and technical assumption~\ref{assum:technical} number 4, we get the upper control 
\begin{equation}
\begin{split}
 u^TH_\theta u &\leq  2\| A \partial_u \phi(\theta) \|_2^2 +  2\|  A\partial_u^2 \phi(\theta)\|_2  (\sqrt{1+\RIPcst} C_{\phi,\theta^*}2\beta+\|e\|_2 ). \\  
 &\leq  2\| A \partial_u \phi(\theta) \|_2^2 +  2M_2 (\sqrt{1+\RIPcst} C_{\phi,\theta^*}2\beta+\|e\|_2 ). \\  
\end{split}
\end{equation}
Using again technical assumptions~\ref{assum:technical} number 4 and  3 give
\begin{equation}
\begin{split}
 \| A \partial_u \phi(\theta) - A \partial_u \phi(\tilde{\theta})\|_2 &\leq  M_2 \| \theta -\tilde{\theta} \|_2 \\
 \| A \partial_u \phi(\theta) \|_2 &\leq \|  A \partial_u \phi(\tilde{\theta})\|_2 + M_2 \| \theta - \tilde{\theta} \|_2\\
 &\leq \|  A \partial_u \phi(\tilde{\theta})\|_2 + M_2 2\beta\\ 
 &\leq M_1 + 2 M_2 \beta.
\end{split}
 \end{equation}
  
We conclude that
\begin{equation}
\begin{split}
 u^TH_\theta  u 
 &\leq 2 \left(M_1 + 2M_2 \beta \right)^2 +2M_2 (\sqrt{1+\RIPcst} C_{\phi,\theta^*}2\beta+\|e\|_2 ) \\  
\end{split}
\end{equation}
and finally
\begin{equation}
 L := \sup_{\theta \in \Lambda_{2\beta}}\sup_{u : \|u\|_2=1}  u^TH_\theta u < +\infty .
\end{equation}

 Hence the  gradient  of $g$ is $L$-Lipschitz on $\Lambda_{2\beta}$. 
 \bigskip
 
 \noindent\textbf{Stability of iterates.} Let $(\theta_n)_{n\geq 0}$ be the sequence of iterates of the gradient descent with fixed step $\tau$ and $\theta_0 \in \Lambda_\beta$. Since $\nabla g$ is $L$-Lipschitz on $\Lambda_{\beta}$ and  $\|\nabla g(\tilde{\theta})\|_2  < \infty$  (using the fact that $\left|\frac{\partial g (\tilde{\theta})}{\partial \theta_i}\right| \leq 2 \|A \frac{\partial \phi (\tilde{\theta})}{\partial \theta_i}\|_2\|e\|_2 \leq 2M_1 \|e\|_2 $),  we have $\sup_{\Lambda_{\beta}} \|\nabla g(\theta)\|_2 < \infty$. We can thus set the descent step
 \begin{equation}
   \tau < \min \left(\frac{1}{L} , \ \frac{\beta}{\sup_{\Lambda_{\beta}} \|\nabla g(\theta)\|_2} \right) .
 \end{equation}
 For each $n$, we introduce $\tilde{\theta}_n \in p(\theta_n, \theta^*)$ which satisfies~\eqref{eq:control_ratio} in order to have the convexity property shown above. Suppose $\theta_n \in \Lambda_\beta$. 
 We then have
 \begin{equation}
   d(\theta_{n+1} , \theta^*) 
   \leq \|\theta_{n+1} - \tilde{\theta}_n \|_2
   \leq \|\theta_n - \tilde{\theta}_n \|_2 +\|\tau \nabla g(\theta_n)\|_2
   < \beta +\tau \sup_{\Lambda_{\beta}} \|\nabla g(\theta)\|_2
   \leq 2\beta.
 \end{equation}
 This proves that $\theta_{n+1} \in \Lambda_{2 \beta}$ and, using technical assumption~\ref{assum:technical} number 1, that $\phi(\theta_{n+1}) \in \Sigma$. These last inequalities also give $[\theta_n, \theta_{n+1}] \subset \Lambda_{2 \beta}$ (by replacing $\tau$ by $\tau'$ such that $0 \leq\tau' \leq \tau$).
 
As shown in the first part of the proof, $g$ is convex on $[\theta_n, \tilde{\theta}_n]$, and thus
  \begin{equation}
\begin{split}
  g(\theta_{n+1}) -  g(\tilde{\theta}_n)
  & =   g(\theta_{n+1}) - g(\theta_n) + g(\theta_n )-  g(\tilde{\theta}_n)\\
& \leq  g(\theta_{n+1}) - g(\theta_n)           +   \ls \nabla g(\theta_n),\theta_n-\tilde{\theta}_n \rs   
\end{split}
  \end{equation}
  
  Since $g$ is twice Gateaux differentiable with $L$-Lipschitz gradient on $\Lambda_{2\beta}$,
  we can use the second-order Taylor inequality on the segment $[\theta_n, \theta_{n+1}] \subset \Lambda_{2\beta}$, which gives
  \begin{equation} \label{eq:descent_inequality}
    g(\theta_{n+1}) - g(\theta_n) \leq  - \tau \ls \nabla g(\theta_n),\nabla g (\theta_n) \rs   +\frac{ L}{2} \|\tau \nabla g(\theta_n) \|_2^2
    =  \left( \frac{ \tau^2L}{2} - \tau \right) \| \nabla g(\theta_n)\|_2^2 .
  \end{equation}
Since $\phi(\theta_{n+1}) \in \Sigma$ and $\tilde{\theta}_n$ is a global minimizer, we get
  \begin{equation}
 \ls \nabla g(\theta_n), \theta_n-\tilde{\theta}_n \rs  -  (\tau - \frac{ \tau^2L}{2} )\| \nabla g(\theta_n)\|_2^2 \geq  g(\theta_{n+1}) - g(\tilde{\theta}_n) \geq 0 .
  \end{equation}

Now,
\begin{equation}
d(\theta_{n+1}, \theta^*)^2 \leq \|\theta_{n+1}-\tilde{\theta}_n\|_2^2 =  \|\theta_{n}-\tilde{\theta}_n\|_2^2 - 2 \ls\theta_{n}-\tilde{\theta}_n , \tau \nabla g(\theta_n) \rs +\tau^2\|\nabla g(\theta_n)\|_2^2.
\end{equation}
Plugging the previous inequality gives
\begin{equation}
\begin{split}
\|\theta_{n+1}-\tilde{\theta}_n\|_2^2&\leq \|\theta_{n}-\tilde{\theta}_n\|_2^2   -  (2\tau^2 -  \tau^3L )\| \nabla g(\theta_n) \|_2^2 +\tau^2\|\nabla g(\theta_n)\|_2^2\\
&= \|\theta_{n}-\tilde{\theta}_n\|_2^2   - c_0 \| \nabla g(\theta_n) \|_2^2 
\end{split}
\end{equation}
where $c_0 = \tau^2(1 - L\tau) > 0$.
Because $\tau < 1/L$, we deduce that $d(\theta_{n+1},\theta^*) < d(\theta_n, \theta^*)$ and $\theta_{n+1} \in \Lambda_{\beta}$.
By induction, we get that the iterates stay in $\Lambda_{\beta}$ because $\theta_0 \in \Lambda_\beta$.
\bigskip

\noindent\textbf{Convergence of $g(\theta_n)$}.
We use the same notation $\theta_n, \tilde{\theta}_n$ as in the last paragraph.
Using again~\eqref{eq:descent_inequality}, we get
\begin{equation}
\begin{split}
  g(\theta_{n+1})-g(\theta^*) &= g(\theta_n) - g(\theta^*) + g(\theta_{n+1}) - g(\theta_n) \\
  &\leq g(\theta_n)- g(\theta^*) -(\tau - \frac{L \tau^2}{2})\|\nabla g(\theta_n)\|_2^2\\
\end{split}
\end{equation}
Remark that $\tau - \frac{L \tau^2}{2} = \tau(1-L\tau/2) \geq \tau/2 >0$.
Using the convexity on $[\theta_n, \tilde{\theta}_n]$ and the Cauchy-Schwarz inequality, we get
\begin{equation}
\begin{split}
g(\theta_{n})- g(\theta^*)=g(\theta_{n})- g(\tilde{\theta}_n)& \leq \ls\nabla g(\theta_n), \theta_n-\tilde{\theta}_n\rs \\
& \leq \|\nabla g(\theta_n)\|_2\| \theta_n-\tilde{\theta}_n \|_2 \\
&\leq \|\nabla g(\theta_n)\|_2\beta.\\
\end{split}
\end{equation}
We get 
\begin{equation}
\begin{split}
g(\theta_{n+1})-g(\theta^*)& \leq g(\theta_n)- g(\theta^*)  -(\tau - \frac{L \tau^2}{2}) \frac{(g(\theta_{n})- g(\theta^*))^2}{\beta^2}\\
& =g(\theta_n)- g(\theta^*)-c_1 (g(\theta_n)- g(\theta^*))^2\\
\end{split}
\end{equation}
with  $c_1 :=(\tau - \frac{L \tau^2}{2}  )  \frac{1}{\beta^2}$. 

Let $d_n =  g(\theta_n)- g(\theta^*)\geq 0$, using $d_{n+1}\leq d_n$, we have that
\begin{equation}
\begin{split}
\frac{d_{n+1}}{d_n} &\leq  1 - c_1 d_n \\
\frac{1}{d_n}& \leq \frac{ 1}{d_{n+1}} - c_1 \frac{d_n}{d_{n+1}} \\
c_1 \leq c_1  \frac{d_n}{d_{n+1}}& \leq \frac{ 1}{d_{n+1}} -\frac{1}{d_n}
\end{split}
\end{equation}
We sum this inequality for  $0, \ldots,n-1$  and get 
\begin{equation}\label{eq:rate_convergence}
\begin{split}
n c_1  & \leq \frac{ 1}{d_n} -\frac{1}{d_0}  \leq \frac{ 1}{d_n} \\
d_n &\leq \frac{1}{c_1 n}.
\end{split}
\end{equation}
\end{proof}

\begin{proof}[Proof of Corollary~\ref{cor:control_general}]
Let $\beta = \min (\beta_1,\beta_2)$, as $\Lambda_{2\beta} \subset \Lambda_{2\beta_1}$, the whole technical assumption~\ref{assum:technical} required for Theorem~\ref{th:control_general} is verified. 

Let $\theta \in \Lambda_\beta \subset \Lambda_{\beta_1}$ and let us consider the unique $\tilde{\theta} \in p(\theta,\theta^*)$.
Using the last hypothesis, we have for all $z \in [\theta,\tilde{\theta}]$ that
 \begin{equation} 
\frac{1}{C_{\phi,\theta^*}}
     \frac{ (1-\RIPcst) \|  \partial_{\tilde{\theta}-\theta} \phi(z)\|_\sH^2 }{\sqrt{1+\RIPcst}\| A\partial_{\tilde{\theta}-\theta}^2 \phi(z)\|_2  } - \frac{1}{C_{\phi,\theta^*} \sqrt{1+\RIPcst}}\|e\|_2 \geq  \beta_2 \geq \beta.
   \end{equation}
   which implies the last hypothesis of Theorem~\ref{th:control_general}
   \begin{equation} 
     \frac{ (1-\RIPcst) \|  \partial_{\tilde{\theta}-\theta} \phi(z)\|_\sH^2 }{\sqrt{1+\RIPcst}\| A\partial_{\tilde{\theta}-\theta}^2 \phi(z)\|_2  }
    \geq C_{\phi,\theta^*} \beta+ \frac{1}{\sqrt{1+\RIPcst}}\|e\|_2 .
   \end{equation}
   
\end{proof}

\begin{proof}[Proof of Corollary~\ref{cor:recovery_guarantees}]
Consider $x_0 \in \Sigma$, and $(\theta_n)_{n\geq 0}$ a sequence of iterates of the gradient descent converging to $\theta^*$ under the hypotheses of Theorem~\ref{th:control_general}. Using the RIP, we have 
\begin{equation}
 \sqrt{1-\gamma}\|\phi(\theta_n) -x_0\|_\sH \leq \|A\phi(\theta_n) -A x_0\|_2 \leq \|A\phi(\theta_n) -y\|_2 + \|e\|_2 
  \end{equation}

  Using the rate of convergence of $g(\theta_n)$ (inequality \eqref{eq:rate_convergence}), we have 
  \begin{equation}
  \begin{split}
    \|A\phi(\theta_n) -y\|_2^2 &\leq \|A\phi(\theta^*) -y\|_2^2 +O\left(\frac{1}{n}\right).\\
     \end{split}
  \end{equation}
  
  This gives, using the inequality $(a+b)^2 \leq 2(a^2 +b^2)$,
 \begin{equation}
 \begin{split}
 \|\phi(\theta_n) -x_0\|_\sH^2 &\leq \frac{1}{1-\gamma}(\sqrt{\|A\phi(\theta^*) -y\|_2^2 +O\left(\frac{1}{n}\right)}+ \|e\|_2 )^2 \\
 &\leq  \frac{2}{1-\gamma}(\|A\phi(\theta^*) -y\|_2^2  +O\left(\frac{1}{n}\right) + \|e\|_2^2 ) \\
 &\leq\frac{2}{1-\gamma}(\|Ax_0 -y\|_2^2 +O\left(\frac{1}{n}\right)+ \|e\|_2^2 ) = \frac{4\|e\|_2^2}{1-\gamma}+O\left(\frac{1}{n}\right).
 \end{split}
 \end{equation}
\end{proof}

\subsection{Proofs for Section~\ref{sec:LR}}\label{sec:proof_LR}
The expression of the directional derivative is deduced from the fact that
\begin{equation}
\begin{split}
 \partial_{ij} ZZ^T &= (\partial_{ij} Z)Z^T + Z(\partial_{ij} Z)^T = (\partial_{ij} Z)Z^T  + ( (\partial_{ij} Z)Z^T)^T \\ 
 &= E_{ij}Z^T + ZE_{ji}\\
\end{split}
\end{equation}
where the family $\{ E_{ij} : 1 \leq i \leq p, 1 \leq j \leq r \}$ is the canonical basis of $\bR^{p \times r}$.
Furthermore
\begin{equation}
\begin{split}
 \partial_{ij} \partial_{U} ZZ^T &= UE_{ij}^T + E_{ij}U^T. \\ 
\end{split}
\end{equation}
%
%
%
\begin{lemma}\label{lem:lower_bound_tr1} 
Let $U \in \bR^{p \times r}$ and $M \in \bR^{r \times r}$. Suppose $M$ PSD, then 
\begin{equation}
\textup{\tr} (U M U^T) \geq \sigma_{min}(M) \|U\|_F^2.
\end{equation}
\end{lemma}
\begin{proof}
 Let $U_i$ be the rows of $U$. We have 
 \begin{equation}
 \begin{split}
  (UMU^T)_{i,i} &= \sum_{l=1,r} U_{i,l} (MU^T)_{l,i} \\
  &=\sum_{l=1,r} U_{i,l}  \sum_{s=1,r} M_{l,s} U_{i,s} \\
   &=\sum_{l=1,r} \sum_{s=1,r}  U_{i,l}   M_{l,s} U_{i,s} \\
   &=  U_{i}   M U_{i}^T \\
   &\geq \sigma_{min}(M)  U_{i}  U_{i}^T \\
   &= \sigma_{min}(M)  \|U_{i}\|_2^2  \\
 \end{split}
 \end{equation}
Summing over $i$ yields the result.
 
\end{proof}

\begin{proof}[Proof of Theorem \ref{th:basin_LR}]
For a rank $r$ matrix $M$, we consider $\sigma_{max}(M) = \sigma_1(M) \geq ... \geq \sigma_r(M)=\sigma_{min}(M)$ the singular values of $M$ in decreasing order. We will use the following inequalities in this proof.
\begin{itemize}
  \item $\sigma_{max}(M)\leq \|M\|_F $
  \item $\|AB\|_F \leq \sigma_{max}(A) \|B\|_F$
\end{itemize}

We verify the hypotheses of Corollary~\ref{cor:control_general}.

\noindent\textbf{Technical assumption~\ref{assum:technical} number 1:} First, remark that for any $Z \in \bR^{p \times r}$, we have $\phi(Z) \in \Sigma_{r} $ (i.e. this assumption is verified for any $\beta$).

\noindent\textbf{Hypothesis 1:} Let $Z \in \Lambda_{\beta_{LR}}$.
Consider $H_0$ the unique solution of the orthogonal Procrustes problem $\min_{H \in \mathcal{O}(r)} \|Z_0H-Z\|_F$. Then it has been shown that $H_0 = QR^T$ where $Q$ and $R$ are obtained from the  singular value decomposition  $Z_0^{T}Z = Q \Delta R^T$  of $Z_0^{T}Z $ (see e.g. \cite[Proof of Lemma~5.7]{Tu_2016}).
It gives 
\begin{equation}
\begin{split}
Z^TZ_0H_0=  R \Delta Q^TQ  R^T = R \Delta R^T  =  RQ^T Q \Delta R^T = H_0^{T} Z_0^{T}Z.
\end{split}
\end{equation}
We get the projection $p(Z,Z_0) = \{Z_0H_0\}$. 

Denoting  $\tilde{Z} := Z_0H_0$, we have $Z^T\tilde{Z}=  \tilde{Z}^{T}Z$ and $ Z^T\tilde{Z}$ PSD. Notice that for $U =\tilde{Z}-Z$, we also have $\tilde{Z}^TU = U^T\tilde{Z}$ and $Z^TU =U^TZ$. Note also that for $1\leq i \leq r$,  $\sigma_i (\tilde{Z}) = \sigma_i (Z_0) $.

 \noindent\textbf{Technical assumption~\ref{assum:technical} number 2:} We calculate a constant $C_{\phi,Z_0}$. Using the triangle inequality on the opertor norm $\sigma_{max}$, the fact that $Z\in \Lambda_{\beta_{LR}} $ and the typical matrix norm inequality $\sigma_{max}(Z-\tilde{Z}) \leq \|Z-\tilde{Z}\|_F$, we have $\sigma_{max}(Z)\leq \sigma_{max}(\tilde{Z}) + \beta_{LR} \leq 2 \sigma_{max}(Z_0)$ and
 \begin{equation}
 \begin{split}
 \| ZZ^T -Z_0Z_0^T \|_F &= \| ZZ^T -\tilde{Z}\tilde{Z}^T \|_F \\
 &= \| Z(Z-\tilde{Z})^T   +  (Z- \tilde{Z})\tilde{Z}^T \|_F\\
 &\leq \| Z(Z-\tilde{Z})^T \|_F  +  \|(Z- \tilde{Z})\tilde{Z}^T \|_F\\
 &\leq (\sigma_{max}(Z)+\sigma_{max}(\tilde{Z})) \| (Z-\tilde{Z})\|_F \\
 &\leq 3 \sigma_{max}(Z_0) \| Z-\tilde{Z}\|_F.
 \end{split}
 \end{equation}
 
Hence we can set $C_{\phi,Z_0} = 3 \sigma_{max}(Z_0)$.   

\noindent\textbf{Technical assumption~\ref{assum:technical} number 3 and 4:} They come from the fact that  $g$ is infinitely differentiable on a bounded domain. 

\noindent\textbf{Hypothesis 3:} First note that, as  $Z \in \Lambda_{\beta_{LR}}$, we have  $\|Z-\tilde{Z}\|_F < \beta_{LR} \leq  \frac{\sigma_{min}(Z_0)}{8}$.  Hence, using Weyl's perturbation inequality, $\sigma_r(Z) \geq \sigma_r(\tilde{Z}) -\sigma_1(Z-\tilde{Z}) \geq \sigma_r(Z_0) -\|Z-\tilde{Z}\|_F  \geq \sigma_r(Z_0) - \frac{\sigma_r(Z_0)}{8} >0 $ and the rank of $Z$ is $r$.

As $\partial_U^2 \phi(Z)$ is an element of $\Sigma_{r}$, with the RIP on $\Sigma_{2r}$, 
\begin{equation}
\begin{split}
 \|A \partial_U^2 \phi(Z)\|_2 =  2\|A UU^T\|_2 &\leq  2\sqrt{1+\RIPcst}\|UU^T\|_F. \\
 \end{split}
\end{equation}

Now,  we  bound the ratio $\frac{\|  \partial_U \phi(Z+tU)\|_F^2 }{ \|A \partial_U^2 \phi(Z)\|_2} = \frac{\| U(Z + t U)^T +(Z + t U)U^T\|_F^2 }{2\|UU^T\|_F}$ for $U= \tilde{Z}-Z$  with $0 \leq t \leq 1$.
 
We have 
\begin{equation}
\begin{split}
\partial_{U} \phi(Z+tU) &=  U(Z + t U)^T +(Z + tU) U^T \\
 &=  UZ^T+Z  U^T  + 2t UU^T . \\
\end{split}
\end{equation}
For any $M_1,M_2 \in \bR^{p \times r}$ with $M_1 \neq 0$, we have that $ \|tM_1 +M_2\|_F^2$ is minimized for $t^* = -\frac{\ls M_1,M_2\rs_F}{\|M_1\|_F^2} $ and  $ \|t^*M_1 +M_2\|_F^2= \|M_2\|_F^2 -\frac{\ls M_1,M_2\rs_F^2}{\|M_1\|_F^2} $. Hence if $t^* \geq 1$, $\|tM_1 +M_2\|_F^2$ is minimized over $[0,1]$ at $t=1$,  at $t=0$ if  $t^* \leq 0$ and $t=t^*$ otherwise. Applying this to the case $M_1= 2UU^T$ and $M_2 = UZ^T+Z  U^T$.

\noindent\textbf{Case 1: $t^* \geq 1$.} 
\begin{equation}
\begin{split}
\min_{t\in [0,1]} \| U(Z + t U)^T +(Z + tU) U^T\|_F^2
&= \|U\tilde{Z}^T +\tilde{Z}U^T \|_F^2.
\end{split}
\end{equation}

\noindent\textbf{Case 2:  $t^* \leq 0$.} 
\begin{equation}
\begin{split}
\min_{t\in [0,1]} \| U(Z + t U)^T +(Z + tU) U^T\|_F^2
&= \|UZ^T +ZU^T \|_F^2.
\end{split}
\end{equation}

 We consider Case 1 and 2 together. Let $\bar{Z} = Z$ or  $\bar{Z} = \tilde{Z}$.  We  have

\begin{equation}
 \begin{split}
\frac{\|U\bar{Z}^T +\bar{Z}U^T \|_F^2}{2\|UU^T\|_F}&=   \frac{\|U\bar{Z}^T\|_F^2 + \|\bar{Z}U^T\|_F^2 + 2 \ls U\bar{Z}^T, \bar{Z}U^T\rs_F }{2\|UU^T\|_F} \\
&=   \frac{2\|\bar{Z}U^T\|_F^2 + 2 \ls \bar{Z}^T, U^T\bar{Z}U^T\rs_F }{2\|UU^T\|_F} .\\
 \end{split}
\end{equation}
Using the fact that $U^T\bar{Z}  = \bar{Z}^TU $,
\begin{equation}
 \begin{split}
\frac{\|U\bar{Z}^T +\bar{Z}U^T \|_F^2}{2\|UU^T\|_F}&=   \frac{\|\bar{Z}U^T\|_F^2 +  \ls \bar{Z}^T, \bar{Z}^TUU^T\rs_F }{\|UU^T\|_F} \\
&=   \frac{\|\bar{Z}U^T\|_F^2 +  \ls \bar{Z}\bar{Z}^T, UU^T\rs_F }{\|UU^T\|_F}. \\
 \end{split}
\end{equation}

Let $A_1,A_2 $ be two positive symmetric matrices (e.g $ZZ^T$ and $UU^T$), then ${\ls A_1,A_2 \rs_F \geq 0}$.
Indeed, let $A_2^{\frac{1}{2}}$ be a positive symmetric square root of $A_2$. Remark that the matrix $ A_2^{\frac{1}{2}} A_1^T A_2^{\frac{1}{2}} = A_2^{\frac{1}{2}} A_1^T (A_2^{\frac{1}{2}})^T$ is positive symmetric.  Hence its trace is positive and we have 
\begin{equation}
 \ls A_1,A_2 \rs_F = \tr (A_1^TA_2) = \tr (A_1^T A_2^{\frac{1}{2}}A_2^{\frac{1}{2}} ) =  \tr (A_2^{\frac{1}{2}} A_1^T A_2^{\frac{1}{2}}) \geq 0.
\end{equation}

We deduce  
\begin{equation}
 \begin{split}
\frac{\|\bar{Z}U^T + U\bar{Z}^T\|_F^2  }{2\|UU^T\|_F} &\geq  \frac{\|\bar{Z}U^T\|_F^2 }{\|UU^T\|_F} =
\frac{\tr (U\bar{Z}^T\bar{Z}U^T)  }{\|UU^T\|_F}. \\
 \end{split}
\end{equation}
But $\bar{Z}^T\bar{Z}$ is a full rank PSD matrix, hence, with Lemma~\ref{lem:lower_bound_tr1},
\begin{equation}
  \tr (U\bar{Z}^T\bar{Z}U^T) \geq \sigma_{min}(\bar{Z}^T\bar{Z}) \tr(UU^T) \geq  (\sigma_{min}(\tilde{Z})-\beta_{LR})^2\tr(UU^T).
\end{equation}
Let $\lambda_i(UU^T)$ denote the eigenvalues of $UU^T$, then 

\begin{equation}\label{eq:eq_LR1} 
 \begin{split}
 \frac{\|  \partial_U \phi(Z+tU)\|_F^2 }{\|\partial_U^2 \phi(Z+tU)\|_F}& \geq \frac{\|\bar{Z}U^T + U\bar{Z}^T\|_F^2 }{2\|UU^T\|_F}\\
&\geq (\sigma_{min}(Z_0)-\beta_{LR})^2 \frac{\sum_i \lambda_i (UU^T) }{\sqrt{\sum_i (\lambda_i(UU^T))^2}} \\
&\geq   (\sigma_{min}(Z_0)-\beta_{LR})^2  \\
&\geq  \left(\frac{7}{8}\sigma_{min}(Z_0)\right)^2 . \\
 \end{split}
\end{equation}

\noindent\textbf{Case 3: $0<t^* <1$.} 

Using the definition of $t^* $, we have 
\begin{equation}\label{eq:proof_LR0}
\begin{split}
 0 &< - \frac{\ls M_1,M_2 \rs_F}{\|M_1\|_F^2} <1  , \\
  0 &< - \frac{\ls UU^T,U\tilde{Z}^T +\tilde{Z}U^T-2UU^T\rs_F}{2\|UU^T\|_F^2} <1 , \\
  0 &< - \frac{\ls UU^T,U\tilde{Z}^T +\tilde{Z}U^T\rs_F}{2\|UU^T\|_F^2} +1 <1 ,  \\
  -1 &< - \frac{\ls UU^T,U\tilde{Z}^T +\tilde{Z}U^T\rs_F}{2\|UU^T\|_F^2}  <0 , \\
    0 &< \ls UU^T,U\tilde{Z}^T +\tilde{Z}U^T\rs_F<2 \|UU^T\|_F^2 . \\
\end{split}
\end{equation}

We calculate
\begin{equation}
 \begin{split}
 \|M_2\|_2^2 - \frac{|\ls M_1,M_2\rs_F|^2}{\|M_1\|_2^2}&= \|UZ^T+ZU^T\|_F^2 - \frac{|\ls UU^T,UZ^T+ZU^T\rs_F|^2}{\|UU^T\|_F^2} \\
&= \|U\tilde{Z}^T +\tilde{Z}U^T-2 UU^T\|_F^2 - \frac{|\ls UU^T,U\tilde{Z}^T +\tilde{Z}U^T-2UU^T\rs_F|^2}{\|UU^T\|_F^2}\\
&= \|U\tilde{Z}^T+\tilde{Z}U^T\|_F^2  -4 \ls UU^T , U\tilde{Z}^T +\tilde{Z}U^T \rs_F +  4\|UU^T\|_F^2 \\
&-  \frac{\left|\ls UU^T,U\tilde{Z}^T +\tilde{Z}U^T\rs_F -2\ls UU^T,UU^T\rs_F \right|^2}{\|UU^T\|_F^2} \\
&= \|U\tilde{Z}^T+\tilde{Z}U^T\|_F^2  -4 \ls UU^T , U\tilde{Z}^T +\tilde{Z}U^T \rs_F +  4\|UU^T\|_F^2 \\
&-  \frac{|\ls UU^T,U\tilde{Z}^T +\tilde{Z}U^T\rs_F|^2}{\|UU^T\|_F^2}  -  4\|UU^T\|_F^2 +4  \ls UU^T,U\tilde{Z}^T +\tilde{Z}U^T\rs_F\\
&=  2\|U\tilde{Z}^T \|_F^2 +2 \ls \tilde{Z}\tilde{Z}^T,UU^T \rs_F  - \frac{|\ls UU^T,U\tilde{Z}^T +\tilde{Z}U^T\rs_F|^2}{\|UU^T\|_F^2} .
 \end{split}
\end{equation}
Using the last inequality of \eqref{eq:proof_LR0}, we have 

\begin{equation}
 \begin{split}
 \|M_2\|_2^2 - \frac{|\ls M_1,M_2\rs_F|^2}{\|M_1\|_2^2} &\geq 2\|U\tilde{Z}^T \|_F^2 +2 \ls \tilde{Z}\tilde{Z}^T,UU^T \rs_F   - 4\|UU^T\|_F^2 .\\
 \end{split}
\end{equation}
Using Equation~\eqref{eq:eq_LR1} and the fact that $\|UU^T\|_F = \sqrt{\sum_i \sigma_i(U^T)^4 }\leq  \sum_{i} \sigma_i(U^T)^2 = \|U^T\|_F^2 \leq \beta_{LR}^2$, we conclude
\begin{equation}
 \begin{split}
 \frac{\|  \partial_U \phi(Z+tU)\|_F^2 }{\|\partial_U^2 \phi(Z+tU)\|_F}& \geq \frac{\|U\tilde{Z}^T \|_F^2 + \ls \tilde{Z}\tilde{Z}^T,UU^T \rs_F   - 2\|UU^T\|_F^2 }{\|UU^T\|_F}\\
&\geq\left(\frac{7}{8}\sigma_{min}(Z_0)\right)^2   - 2\|UU^T\|_F \\
&\geq  \left(\frac{7}{8}\sigma_{min}(Z_0)\right)^2   - 2 \beta_{LR}^2\\
&\geq   (\frac{49}{64}   - \frac{2}{64})(\sigma_{min}(Z_0))^2 = \frac{47}{64} (\sigma_{min}(Z_0))^2 >0.\\
 \end{split}
\end{equation}

This gives, for $\beta_2$ from Corollary~\ref{cor:control_general},
 \begin{equation}
 \begin{split}
  \beta_2 &= \frac{1}{C_{\phi,Z_0}}\inf_{Z \in \Lambda_{\beta_1}} \inf_{Y \in [Z, \tilde{Z}]}\left(
     \frac{ (1-\RIPcst) \|  \partial_{\tilde{Z}-Z } \phi(Y)\|_\sH^2 }{\sqrt{1+\RIPcst}\| A\partial_{\tilde{Z}-Z } ^2 \phi(Y)\|_2  }     \right) \\
   &\geq \frac{47}{64} \frac{  (1-\RIPcst)  (\sigma_{min}(Z_0))^2 }{3\sigma_{max}(Z_0)(1+\RIPcst)}\\
   &\geq \frac{1}{8} \frac{  (1-\RIPcst)  (\sigma_{min}(Z_0))^2 }{(1+\RIPcst)\sigma_{max}(Z_0)}= \beta_{LR} >0 ,
  \end{split}
 \end{equation}
which implies $\Lambda_{\beta_{LR}} \subset \Lambda_{\min (\beta_1,\beta_2)}$. Using Corollary~\ref{cor:control_general} yields the final result: the set $\Lambda_{\beta_{LR}}$ is a $g$-basin of attraction of $\theta^*$.

\end{proof}

 \subsection{Proofs for the GMM example} \label{sec:proof_GMM}
 We start by giving two lemma that bound derivatives of $\phi$.
 \begin{lemma}\label{lem:bound_kernel_norm_gaussian}
  Suppose $K(t) \propto e^{-\frac{1}{2}\lambda^2\|t\|_\Gamma^2}$. Then  there is a strictly positive constant $d_K$ such that $\|\partial_{w} \mu_{t_0}\|_K^2 \geq d_K \|w\|_2^2$ where $\partial_{w} \mu_{t_0}$ is the derivative of $\mu_{t_0}$ with respect to $t$ in the direction $w$. 
 \end{lemma}
 \begin{proof}
  Using Equation~\eqref{eq:GMM_kernel_norm_expr1},
 \begin{equation}
 \begin{split}
 \|\partial_{w} \mu_{t_0}\|_K^2&\propto \int_{\bR^p}\int_{\bR^p} e^{-\frac{1}{2}\lambda^2\|t-s\|_\Gamma^2}w^T\Gamma^{-1}s w^T\Gamma^{-1}te^{-\frac{1}{2}\|s-t_0\|_\Gamma^2}e^{-\frac{1}{2}\|t-t_0\|_\Gamma^2} \id t \id s \\ 
 \end{split}
 \end{equation}
 Using the change of variables $t -t_0 \to t$ and $s-t_0 \to s$, we get 
 \begin{equation}
 \begin{split}
 \|\partial_{w} \mu_{t_0}\|_K^2&\propto\int_{\bR^p}\int_{\bR^p} e^{-\frac{1}{2}\lambda^2\|t-s\|_\Gamma^2}w^T\Gamma^{-1}s w^T\Gamma^{-1}te^{-\frac{1}{2}\|s\|_\Gamma^2}e^{-\frac{1}{2}\|t\|_\Gamma^2} \id t \id s \\ 
  \end{split}
 \end{equation}
 
 Rewrite 
 \begin{equation}\label{eq:change_var_GMM1}
 \begin{split} 
  \lambda^2\|s-t\|_\Gamma^2 + \|s\|_\Gamma^2 &= (1+\lambda^2) \|s\|_\Gamma^2 - 2\lambda^2 \ls s,t\rs +\lambda^2\|t\|_\Gamma^2\\
  &= (1+\lambda^2) \left( \|s\|_\Gamma^2 - 2\frac{\lambda^2}{1+\lambda^2} \ls s,t\rs +\frac{\lambda^2}{1+\lambda^2}\|t\|_\Gamma^2 \right) \\
  &= (1+\lambda^2) \left( \left\|s- \frac{\lambda^2}{1+\lambda^2} t \right\|_\Gamma^2  +(1-\frac{\lambda^2}{1+\lambda^2})\frac{\lambda^2}{1+\lambda^2}\|t\|_\Gamma^2\right)\\
    &= (1+\lambda^2)  \left\|s- \frac{\lambda^2}{1+\lambda^2} t\right\|_\Gamma^2  + \frac{\lambda^2}{1+\lambda^2}\|t\|_\Gamma^2.\\
  \end{split}
 \end{equation}
 This leads to
 \begin{equation}
 \begin{split} 
  \|\partial_{w} \mu_{t_0}\|_K^2&\propto \int_{\bR^p}\int_{\bR^p} w^T\Gamma^{-1}s w^T\Gamma^{-1}te^{-\frac{1}{2}(1+\lambda^2)\|s- \frac{\lambda^2}{1+\lambda^2}t\|_\Gamma^2}e^{-\frac{1}{2}\left(\|t\|_\Gamma^2+ \frac{\lambda^2}{1+\lambda^2} \|t\|_\Gamma^2\right) } \id t \id s \\
    &= \int_{\bR^p}\int_{\bR^p} w^T\Gamma^{-1}s w^T\Gamma^{-1}te^{-\frac{1}{2}(1+\lambda^2)\|s- \frac{\lambda^2}{1+\lambda^2}t\|_\Gamma^2}e^{- \frac{1}{2}\frac{1+2\lambda^2}{1+\lambda^2} \|t\|_\Gamma^2 } \id t \id s .
\end{split}
 \end{equation}
 Let $h(s) = e^{-\frac{1}{2}(1+\lambda^2)\|s- \frac{\lambda^2}{1+\lambda^2}t\|_\Gamma^2}$,
 \begin{equation}
 \begin{split}
 \int_s w^T\Gamma^{-1}s e^{-\frac{1}{2}(1+\lambda^2)\|s-  \frac{\lambda^2}{1+\lambda^2}t\|_\Gamma^2}\id s &= \int_s w^T\Gamma^{-1} (s- \frac{\lambda^2}{1+\lambda^2}t)h(s)\id s  \\
&+ \frac{\lambda^2}{1+\lambda^2}w^T\Gamma^{-1}t \int_s  h(s)\id s \\
 &= \int_s \partial_{w}h(s)\id s  + \frac{\lambda^2}{1+\lambda^2} w^T\Gamma^{-1}t \mathbf{C}\left(\frac{\Gamma}{1+\lambda^2}\right)
   \end{split}
 \end{equation}
 where $\mathbf{C}\left(\frac{\Gamma}{1+\lambda^2}\right)$ is the normalization constant of the Gaussian of covariance  $\frac{\Gamma}{1+\lambda^2}$. By linearity 
 \begin{equation}
 \begin{split}
  \int_s \partial_{w}h(s)\id s &= \ls w, \int_s \nabla h(s)\id_s \rs = 0.
   \end{split}
 \end{equation}
We make the change of variable $ u =\Gamma^{-1/2}t$
 \begin{equation}
 \begin{split}
 \|\partial_{w} \mu_{t_0}\|_K^2&\propto   \frac{\lambda^2}{1+\lambda^2} \mathbf{C}\left(\frac{\Gamma}{1+\lambda^2}\right) |\det (\Gamma^{-1/2})|  \int_u | w^T\Gamma^{-1/2}u|^2  e^{-\frac{1}{2} \frac{1+2\lambda^2}{1+\lambda^2}\|u\|_2^2} \id u .
   \end{split}
 \end{equation}
With the linearity of the integral, the change of variable $u \to \sqrt{\frac{1+\lambda^2}{1+2\lambda^2}}u$ and the fact that $\int_\bR x^2e^{-\frac{1}{2}x^2} \id x = \sqrt{2 \pi}$,
 \begin{equation}
 \begin{split}
 \|\partial_{w} \mu_{t_0}\|_K^2&\propto  |\det (\Gamma^{-1/2})|  \frac{\lambda^2}{1+\lambda^2} \mathbf{C}\left(\frac{\Gamma}{1+\lambda^2}\right) \sum_i \int_u |u_i (\Gamma^{-1/2}w)_i|^2 e^{-\frac{1}{2} \frac{1+2\lambda^2}{1+\lambda^2}\|u\|_2^2} \id t \\
 & + \sum_{i\neq j} \int_u u_i (\Gamma^{-1/2}w)_i u_j (\Gamma^{-1/2}w)_j e^{-\frac{1}{2} \frac{1+2\lambda^2}{1+\lambda^2}\|u\|_2^2} \id u) \\
&=  |\det (\Gamma^{-1/2})|   \frac{\lambda^2}{1+\lambda^2} \mathbf{C}\left(\frac{\Gamma}{1+\lambda^2}\right) \left( \sqrt{\frac{1+\lambda^2}{1+2\lambda^2}}\right)^{p}\\
& \times \sum_i \int_u  \frac{1+\lambda^2}{1+2\lambda^2} |u_i (\Gamma^{-1/2}w)_i|^2 e^{-\frac{1}{2} \|u\|_2^2} \id u \\
&=  |\det (\Gamma^{-1/2})|   \frac{\lambda^2}{1+2\lambda^2} \mathbf{C}\left(\frac{\Gamma}{1+\lambda^2}\right) \left( \sqrt{\frac{1+\lambda^2}{1+2\lambda^2}}\right)^{p} \sqrt{2 \pi}^{p} \sum_i |(\Gamma^{-1/2}w)_i|^2  \\
&=  |\det (\Gamma^{-1/2})|   \frac{\lambda^2}{1+2\lambda^2} \mathbf{C}\left(\frac{\Gamma}{1+\lambda^2}\right)  \left( \sqrt{\frac{1+\lambda^2}{1+2\lambda^2}}\right)^{p}\sqrt{2 \pi}^{p} \|\Gamma^{-1/2}w\|_2^2  \\
&\geq  |\det (\Gamma^{-1/2})|   \frac{\lambda^2}{1+2\lambda^2} \mathbf{C}\left(\frac{\Gamma}{1+\lambda^2}\right) \left( \sqrt{\frac{1+\lambda^2}{1+2\lambda^2}}\right)^{p}\sqrt{2 \pi}^{p}\lambda_{min}(\Gamma) \|w\|_2^2  .
   \end{split}
 \end{equation}
 This gives the result with ${d_K \propto  |\det (\Gamma^{-1/2})|   \frac{\lambda^2}{1+2\lambda^2} \mathbf{C}\left(\frac{\Gamma}{1+\lambda^2}\right)  \left( \sqrt{\frac{1+\lambda^2}{1+2\lambda^2}}\right)^{p}\sqrt{2 \pi}^{p}\lambda_{min}(\Gamma) > 0}$.
 \end{proof}


 \begin{lemma}\label{lem:bound_kernel_norm_gaussian2}
  Suppose $K(t) \propto e^{-\frac{1}{2}\lambda^2\|t\|_\Gamma^2}$. Then  there is an explicit strictly positive constants $D_K$ depending on $K$  such that 
  \begin{equation}
  \|\partial_{w} \mu_{t_0}\|_K^2 \leq  D_K\|w\|_2^2
  \end{equation}
  and there is an explicit strictly positive constants $D_{A,K}'$ depending on  $A$ and $K$  such that 
  \begin{equation}
  \|A\partial_{w}^2 \mu_{t_0}\|_2 \leq D_{A,K}'(\|w\|_2^2+\|w\|_2)
  \end{equation}
  where $\partial_{w}^2 \mu_{t_0}$ is the second derivative of $\mu_{t_0}$ with respect to $t$ in the direction $w$ and
  \begin{equation}
  \begin{split}
D_K &:=   \int_{s\in\bR^p}\int_{t\in\bR^p} K(t-s) \|\Gamma^{-1}s\|_2 \|\Gamma^{-1}t\|_2 e^{-\frac{1}{2}\|s\|_\Gamma^2}e^{-\frac{1}{2}\|t\|_\Gamma^2}  \id t \id s \\
D_{A,K}'&:=  \sqrt{\sum_{l=1}^{m}  \sup_{t\in \bR^p} |\alpha_l(t)|^2} \max \left( \frac{1}{\lambda_{min}(\Gamma)}\mathbf{C}\left(\Gamma\right), \int_{t\in \bR^p} \| \Gamma^{-1}t\|_2e^{-\frac{1}{2}\|t\|^2_\Gamma}\right).
\end{split}
 \end{equation}
 \end{lemma}
\begin{proof}

 For the first bound, we have 
 \begin{equation}
  \begin{split}
 \|\partial_{w} \mu_{t_0}\|_\sH^2 &= \int_{s\in\bR^p}\int_{t\in\bR^p} K(t-s) w^T\Gamma^{-1}s w^T\Gamma^{-1}t e^{-\frac{1}{2}\|s\|_\Gamma^2}e^{-\frac{1}{2}\|t\|_\Gamma^2}  \id t \id s \\
 &\leq \|w\|_2^2 \int_{s\in\bR^p}\int_{t\in\bR^p} K(t-s) \|\Gamma^{-1}s\|_2 \|\Gamma^{-1}t\|_2 e^{-\frac{1}{2}\|s\|_\Gamma^2}e^{-\frac{1}{2}\|t\|_\Gamma^2}  \id t \id s \\
 &=  D_K \|w\|_2^2
  \end{split}
 \end{equation}
 where $ D_K :=   \int_{s\in\bR^p}\int_{t\in\bR^p} K(t-s) \|\Gamma^{-1}s\|_2 \|\Gamma^{-1}t\|_2 e^{-\frac{1}{2}\|s\|_\Gamma^2}e^{-\frac{1}{2}\|t\|_\Gamma^2}  \id t \id s $.
 
For the second bound,  we have (with $\mathbf{C}\left(\Gamma\right)$ the normalization constant of Gaussian of covariance $\Gamma$),
 \begin{equation}
  \begin{split}
 \|A\partial_{w_i}^2 \mu_{t_0} \|_2^2&=  \sum_{l=1}^{m} \left|\int \alpha_l(t)\id \partial_{w}^2 \mu_{t_0}(t) \right|^2 \\
 &\leq   \left(\sum_{l=1}^{m}  \sup_{t\in \bR^p} |\alpha_l(t)|^2\right) \left(\int_{t\in \bR^p} \left|\id \partial_{w}^2 \mu_{t_0}(t) \right|\right)^2 \\
 &=   \left(\sum_{l=1}^{m}  \sup_{t\in \bR^p} |\alpha_l(t)|^2\right) \left(\int_{t\in \bR^p}\left|( - w^T \Gamma^{-1}w + w^T \Gamma^{-1}(t-t_0) )e^{-\frac{1}{2}\|t-t_0\|^2_\Gamma} \right|\id t\right)^2.\\
  \end{split}
 \end{equation}
 This gives
 \begin{equation}
  \begin{split}
 \|A\partial_{w_i}^2 \mu_{t_0} \|_2 &\leq  \sqrt{\sum_{l=1}^{m}  \sup_{t\in \bR^p} |\alpha_l(t)|^2} \left(\mathbf{C}\left(\Gamma\right) \|w\|_\Gamma^2 + \int_{t\in \bR^p} \left|w^T \Gamma^{-1}(t-t_0)e^{-\frac{1}{2}\|t-t_0\|^2_\Gamma} \id t\right|\right)\\
 &\leq  \sqrt{\sum_{l=1}^{m}  \sup_{t\in \bR^p} |\alpha_l(t)|^2} \left(\mathbf{C}\left(\Gamma\right) \|w\|_\Gamma^2 + \|w\|_2\int_{t\in \bR^p} \| \Gamma^{-1}t\|_2e^{-\frac{1}{2}\|t\|^2_\Gamma} \id t \right)\\
 &\leq  \sqrt{\sum_{l=1}^{m}  \sup_{t\in \bR^p} |\alpha_l(t)|^2} \left(\mathbf{C}\left(\Gamma\right) \frac{1}{\lambda_{min}(\Gamma)}\|w\|_2^2 + \|w\|_2\int_{t\in \bR^p} \| \Gamma^{-1}t\|_2e^{-\frac{1}{2}\|t\|^2_\Gamma} \id t \right) \\
 &\leq  D_{A,K}' \left(\|w\|_2^2  +\|w\|_2 \right)\\
  \end{split}
 \end{equation}
where $D_{A,K}':=  \sqrt{\sum_{l=1}^{m}  \sup_{t\in \bR^p} |\alpha_l(t)|^2} \max \left( \frac{1}{\lambda_{min}(\Gamma)}\mathbf{C}\left(\Gamma\right), \int_{t\in \bR^p} \| \Gamma^{-1}t\|_2e^{-\frac{1}{2}\|t\|^2_\Gamma}\right)$.
\end{proof}

\begin{proof} [Proof of Theorem \ref{th:control_gaussian}]
 We prove this theorem by verifying the hypotheses of Corollary~\ref{cor:control_general}.  We  take $\beta_1 = \beta_{GMM} \leq \frac{\sqrt{\lambda_{min}(\Gamma)} \epsilon}{8}$ and $\Sigma = \Sigma_{k,\frac{\epsilon}{2}}$. We recall that we write $\theta^* = (a_1,...,a_k,t_1,...,t_k)$.

 \noindent\textbf{Technical assumption~\ref{assum:technical} number 1}: Let $\theta = (b_1,...,b_k,s_1,...,s_k)\in \Lambda_{2\beta_1}$. Similarly to~\cite{Traonmilin_2019a}, $\|s_i-s_j\|_\Gamma = \|s_i -t_i +t_i -t_j +t_j-s_j\|_\Gamma \geq \|t_i -t_j\|_\Gamma -\|t_i-s_i\|_\Gamma -\|t_j-s_j\|_\Gamma \geq \epsilon - \frac{1}{\sqrt{\lambda_{min}(\Gamma)}}(\|t_i-s_i\|_2+\|t_j-s_j\|_2 ) > \epsilon - 2\epsilon/4 \geq \epsilon/2$. Moreover, using the hypothesis on the $t_i$, $\|s_i\|_2 \leq \|t_i\|_2 + 2\beta_{GMM} \leq R$ and $\phi(\theta) \in \Sigma_{k,\frac{\epsilon}{2}}$.
 
 \noindent\textbf{Hypothesis 1:} The RIP and the properties of $K$ guarantee the unicity of $ \phi(\theta^*)$ as a minimizer of~\eqref{eq:minimization}. Now for $\theta \in \Lambda_{2\beta_1}$,  the set $p(\theta,\theta^*)$ is included in the set of all the possible orderings of amplitudes and positions. The fact that $\|\cdot\|_K$ increases  with respect to the distance between positions and that $\|t_i-s_i\|_\Gamma\leq \frac{\epsilon}{4} $ imply $p(\theta,\theta^*) = \{\theta^*\}$.

\noindent\textbf{Technical assumption~\ref{assum:technical} number 2:} This hypothesis comes  from the following. For  $\sum_i a_i \mu_{t_i} - \sum_i b_i\mu_{s_i} \in \Sigma -\Sigma$,
\begin{equation}
 \left\|\sum_i a_i \mu_{t_i} - \sum_i b_i\mu_{s_i}\right\|_K \leq \sum_i\left\|a_i \mu_{t_i} -  b_i\mu_{s_i}\right\|_K. 
\end{equation}
Moreover,
\begin{equation}
\|a_i \mu_{t_i} -  b_i\mu_{s_i}\|_K^2 = (a_i-b_i)^2 + 2a_ib_i(1-\ls \mu_{t_i},\mu_{s_i} \rs_K).
\end{equation}
 It was shown in~\cite{Gribonval_2017} that for a well designed Gaussian kernel, there is an explicit constant $C_K$   such that $1-\ls \mu_{t_i},\mu_{s_i} \rs_K \leq C_K \|t_i-s_i\|_2^2$ for $\|t_i-s_i\|_{\Gamma}^2 \leq \epsilon/4$.   Hence for $\theta \in \Lambda_{2\beta_1}$, using the fact that $|b_i| \leq |a_i| + 2\beta_1 \leq 2 |a_i| $, we have  
\begin{equation}
\begin{split}
\|\phi(\theta)-\phi(\theta^*)\|_K^2 &\leq \sum_i (a_i-b_i)^2 + 4|a_i|^2 C_K \|t_i-s_i\|_2^2\\
 & \leq \max(1,4C_Ka_k^2) \|\theta-\theta^*\|_2^2\\
\end{split}
\end{equation}

 Hence we can set the constant  $C_{\phi,\theta^*} = \max(1,4C_Ka_k^2) $.
 
 \noindent\textbf{Technical assumption~\ref{assum:technical} number 3 and 4:} These hypotheses come from the fact that the function $g$ is infinitely differentiable on the bounded domain $\Lambda_{2\beta_{1}}$.

  \noindent\textbf{Hypothesis 3:} Let $\theta \in \Lambda_{\beta_{GMM}}$  such that $\phi(\theta) =\sum_{i=1}^k b_i\mu_{s_i}$ and $u=(v,w) \in \bR^d$ such that $\|u\|_2 =1$. We have, using the kernel assumption,
 \begin{equation}
 \begin{split}
 \|  \partial_u \phi(\theta)\|_K^2 &\geq (1-c_K) \sum_i \|v_i \mu_{s_i} + b_i \partial_{w_i}\mu_{s_i}\|_K^2\\
 &= (1-c_K) \sum_i \left( |v_i|^2 \|\mu_{s_i}\|_K^2 + |b_i|^2 \|\partial_{w_i}\mu_{s_i}\|_K^ 2 + 2b_iv_i\ls\mu_{s_i}, \partial_{w_i} \mu_{s_i}\rs_K \right) .\\
 \end{split}
 \end{equation}
 We calculate the cross-product
 \begin{equation}
 \begin{split}
  \ls\mu_{s_i}, \partial_{w_i} \mu_{s_i}\rs_K &=  -2\int_{\bR^p}\int_{\bR^p}  K(t,s) w_i^T\Gamma^{-1}(t-s_i)e^{-\frac{1}{2}\|s-s_i\|_\Gamma^2}e^{-\frac{1}{2}\|t-s_i\|_\Gamma^2} \id t \id s. \\
\end{split}
 \end{equation}
The kernel $K$ is written  $K(t,s) \propto e^{-\lambda^2\|t-s\|_\Gamma^2}$. With the translational invariance,
 \begin{equation}
 \begin{split}
  \ls\mu_{s_i}, \partial_{w_i} \mu_{s_i}\rs_K &\propto  \int_{\bR^p}\int_{\bR^p}  e^{-\frac{1}{2}\lambda^2\|t-s\|_\Gamma^2} w_i^T\Gamma^{-1}(t-s_i)e^{-\frac{1}{2}\|s-s_i\|_\Gamma^2}e^{-\frac{1}{2}\|t-s_i\|_\Gamma^2} \id t \id s \\
  &=  \int_{\bR^p}\int_{\bR^p}  e^{-\frac{1}{2}\lambda^2\|t-s\|_\Gamma^2} w_i^T\Gamma^{-1}te^{-\frac{1}{2}\|s\|_\Gamma^2}e^{-\frac{1}{2}\|t\|_\Gamma^2} \id t \id s. \\
\end{split}
 \end{equation}
 Using identity~\eqref{eq:change_var_GMM1} gives 
 \begin{equation}
 \begin{split}
  \int  e^{-\frac{1}{2}\lambda^2\|t-s\|_\Gamma^2} e^{-\frac{1}{2}\|s\|_\Gamma^2} \id s &= \mathbf{C}\left(\frac{\Gamma}{1+\lambda^2}\right)e^{-\frac{1}{2}\frac{\lambda^2}{1+\lambda^2}\|t\|_\Gamma^2}.
\end{split}
 \end{equation}
 where $\mathbf{C}  (X):= \int_{\bR^p} e^{-\frac{1}{2}\|s\|_X^2}\id s$. 
 Using the fact that $t \to h(t) = w^T\Gamma^{-1}t e^{-\frac{1}{2}\|t\|_\Gamma^2}e^{-\frac{1}{2}\frac{\lambda^2}{1+\lambda^2}\|t\|_\Gamma^2}$ is an odd function of $t$, the integral with respect to $t$ is zero and 
 \begin{equation}
 \begin{split}
  \ls\mu_{s_i}, \partial_{w_i} \mu_{s_i}\rs_K =0.
\end{split}
 \end{equation}

 Hence, using the  assumption that $\|\mu_{s_i}\|_K^2=1 $ and $\|(v,w)\|_2 =1$,
 \begin{equation}
 \begin{split}
 \|  \partial_u \phi(\theta)\|_K^2  &\geq (1-c_K) \left( \sum_i |v_i|^2 \|\mu_{s_i}\|_K^2 + |b_i|^2 \|\partial_{w_i}\mu_{s_i}\|_K^2 \right) \\
  &\geq (1-c_K) \left(\sum_i |v_i|^2 + |b_i|^2 d_K \|w_i\|_2^2 \right) \\
  &\geq (1-c_K) \min( 1 , d_K \min_i (|b_i|^2)) .
  \end{split}
 \end{equation}
 where $d_K$ such that $  \|\partial_{w} \mu_{0}\|_K^2 \geq d_K\|w\|_2^2$ is given by Lemma~\ref{lem:bound_kernel_norm_gaussian}.
 
 We now bound 
 \begin{equation}
  \begin{split}
 \|A\partial_u^2 \phi(\theta)\|_\sH  &= \| \sum_i A\partial_u^2 b_i \mu_{s_i}\|_\sH. \\
  \end{split}
 \end{equation}
We have 
\begin{equation}
\begin{split}
  \partial_u^2 b_i \mu_{s_i} &=\partial_u (v_i \mu_{s_i} + b_i \partial_{w_i} \mu_{s_i}) = v_i\partial_{w_i} \mu_{s_i} +  b_i \partial_{w_i}^2 \mu_{s_i}  + v_i \partial_{w_i} \mu_{s_i}\\
  &= 2v_i\partial_{w_i} \mu_{s_i} +  b_i \partial_{w_i}^2 \mu_{s_i}. \\
  \end{split}
\end{equation}
Hence  
 \begin{equation}
  \begin{split}
 \|A\partial_u^2 \phi(\theta)\|_\sH  
 & = \| \sum_i 2A v_i \partial_{w_i} \mu_{s_i} + b_i\partial_{w_i}^2 \mu_{s_i} \|_2\\
 & \leq    2\sqrt{1+\gamma} \|\sum_i v_i \partial_{w_i} \mu_{s_i}\|_\sH  + \|\sum_i b_iA\partial_{w_i}^2 \mu_{s_i} \|_2 \\
 & \leq 2\sum_i   \sqrt{1+\gamma}\sqrt{1+c_k} \sqrt{\sum_i v_i^2\|\partial_{w_i} \mu_{s_i}\|_\sH^2}  + |b_k|\sum_i \|A\partial_{w_i}^2 \mu_{s_i} \|_2 \\
  \end{split}
 \end{equation}
 With Lemma~\ref{lem:bound_kernel_norm_gaussian2}, we have
 \begin{equation}
  \begin{split}
  \|A\partial_u^2 \phi(\theta)\|_\sH   & \leq  2\sqrt{1+\gamma}\sqrt{1+c_k} \sqrt{D_K}\|v\|_2  + |b_k|D_{A,K}' \sum_i(\|w_i\|_2^2 +\|w_i\|_2  )\\
   & \leq  2\sqrt{1+\gamma}\sqrt{1+c_k} \sqrt{D_K}\|v\|_2  + |b_k|D_{A,K}' (\|w\|_2^2 + \|w\|_2) \\
   & \leq 2 \sqrt{1+\gamma}\sqrt{1+c_k} \sqrt{D_K} + 2|b_k|D_{A,K}'.\\
  \end{split}
 \end{equation}

 Using the fact that $\beta_{GMM} \leq \frac{|a_1|}{2}$, we have
$|b_i| \geq |a_i| - \beta_{GMM} \geq |a_1| - \beta_{GMM} \geq \frac{|a_1|}{2}$ and $|b_k| \leq 2 |a_k|$. This gives, for $\beta_2$ from Corollary~\ref{cor:control_general},
 \begin{equation}
 \begin{split}
  \beta_2 &= \frac{1}{C_{\phi,\theta^*}}\inf_{\theta \in \Lambda_{\beta_1}} \inf_{z \in [\theta, \theta^*]}\left(
     \frac{ (1-\RIPcst) \|  \partial_{\theta^*-\theta} \phi(z)\|_\sH^2 }{\sqrt{1+\RIPcst}\| A\partial_{\theta^*-\theta}^2 \phi(z)\|_2  }     \right) \\
     & \geq \frac{1}{C_{\phi,\theta^*}}\inf_{\theta \in \Lambda_{\beta_1}} \inf_{u : \|u\|_2=1}\left(
     \frac{ (1-\RIPcst) \|  \partial_{u} \phi(\theta)\|_\sH^2 }{\sqrt{1+\RIPcst}\| A\partial_{u}^2 \phi(\theta)\|_2  }     \right) \\
   &\geq \frac{  (1-\RIPcst)  (1-c_K) \min(1 , d_K |a_1|^2) }{8C_{\phi,\theta^*}\sqrt{1+\RIPcst}(\sqrt{1+\gamma}\sqrt{1+c_k} \sqrt{D_K} + 2|a_k|D_{A,K}')} >0 
  \end{split}
 \end{equation}

and $\beta_{GMM} \leq \beta_2$ which implies $\Lambda_{\beta_{GMM}} \subset \Lambda_{\min (\beta_1,\beta_2)}$. Finally the set $\Lambda_{\beta_{GMM}}$ is a $g$-basin of attraction of $\theta^*$.

\end{proof}
 \subsubsection{Proofs for GMM with variable covariances} \label{sec:proof_GMM_var}
 We write $Z_{i,k,l}$ the coordinate $k,l$ of matrix $Z_i$. We have that $\partial_{Z_{i,k,l}} \mu_{t_i,\Gamma_i}$ has density  
 \begin{equation}
  -\frac{1}{2} (\partial_{Z_{i,k,l}}  \|t-t_i\|_{(Z_iZ_i^T+\rho I) }^2) e^{-\frac{1}{2}\|t-t_i\|_{\Gamma_i}^2}.  \\
 \end{equation}
 We also have, using the inverse matrix differentiation formula $\partial(A^{-1}) = -A^{-1} \partial(A)A^{-1}$
\begin{equation}
\begin{split}
\partial_{Z_{i,k,l}}  \|t-t_i\|_{(Z_iZ_i^T+\rho I) }^2&=  - (t-t_i)^T (Z_iZ_i^T +\rho I)^{-1}  \partial_{Z_{i,k,l}} (Z_iZ_i^T)(Z_iZ_i^T +\rho I)^{-1}(t-t_i).\\
\end{split}
\end{equation}
 Using the fact that $\partial_{W} (Z_iZ_i^T) = WZ_i^T +Z_iW^T$, we get the expression of the directional derivative $\partial_{W} \mu_{t_i,\Gamma_i}$.\\

 \begin{proof}[Proof of Lemma \ref{lem:bound_variable_cov}]
We write $(X_i)_{i=1}^p$ the diagonal terms of a matrix $X \in \bR^{p \times p}$.
With the translational invariance of $K$ we have $\|\partial_{W} \mu_{t,\Gamma}\|_K^2   = \|\partial_{W} \mu_{0,\Gamma}\|_K^2 $.
 
 The partial derivative in direction $W$ of $\mu_{0,\Gamma}$ has density
 \begin{equation}
 \begin{split}
  h(s)&=\frac{1}{2}s^T \Gamma^{-1}W\Gamma^{-1} s e^{-\frac{1}{2}\|s\|_{\Gamma}^2}\\
  &=\frac{1}{2}s^T \Gamma^{-2}Ws e^{-\frac{1}{2}\|s\|_{\Gamma}^2} ,
  \end{split}
  \end{equation}
  because diagonal matrices commute, and thus
  \begin{equation}
 \begin{split}
 \|\partial_{W} \mu_{0,\Gamma}\|_K^2   & = \int_{\bR^p}\int_{\bR^p} K(s,t)h(s)h(t)\id s \id t .
\end{split}
 \end{equation}
The aim of the following computations is to give a lower bound on $\|\partial_{W} \mu_{0,\Gamma}\|_K^2$.

The rest of the proof is quite technical, and we thus split it into four steps for ease of reading.

\noindent \textbf{Step 1: rewriting $\|\partial_{W} \mu_{0,\Gamma}\|_K^2$.}

  We have 
  \begin{equation}
 \begin{split}
 \|s\|_\Gamma^2 + \lambda\|s-t\|_2^2   & =  s^T(\Gamma^{-1}+\lambda I)s  -2 \lambda \ls s,t\rs  + \lambda\|t\|_2^2 \\
 &= s^T(\Gamma^{-1}+\lambda I)s  -2  s^T (\Gamma^{-1}+\lambda I) (\Gamma^{-1}+\lambda I)^{-1} \lambda t  + \lambda\|t\|_2^2 \\
&=  \|s - \lambda (\Gamma^{-1}+\lambda I)^{-1}  t\|_{(\Gamma^{-1}+\lambda I)^{-1}}^2    + \lambda\|t\|_2^2 - \|\lambda (\Gamma^{-1}+\lambda I)^{-1}  t\|_{(\Gamma^{-1}+\lambda I)^{-1}}^2 .
\end{split}
  \end{equation}
 Hence  
  \begin{equation}\label{eq:varcov1}
  \begin{split}
 \int_{s \in \bR^p} K(s,t)h(s) \id s =& \frac{1}{2}e^{-\frac{1}{2}\lambda\|t\|_2^2 +\frac{1}{2} \|\lambda (\Gamma^{-1}+\lambda I)^{-1}  t\|_{(\Gamma^{-1}+\lambda I)^{-1}}^2}\\
&\times \int_{s \in \bR^p} s^T \Gamma^{-2}Ws e^{- \frac{1}{2}\|s - \lambda (\Gamma^{-1}+\lambda I)^{-1}  t\|_{(\Gamma^{-1}+\lambda I)^{-1}}^2  } \id s .
\end{split}
  \end{equation}
 We calculate, using the change of variable $ s - \lambda (\Gamma^{-1}+\lambda I)^{-1}  t \to s $,
  \begin{equation}
 \begin{split}
 B&:=\int_{s \in \bR^p} s^T \Gamma^{-2}Ws e^{- \frac{1}{2}\|s - \lambda (\Gamma^{-1}+\lambda I)^{-1}  t\|_{(\Gamma^{-1}+\lambda I)^{-1}}^2  } \id s \\
 &= \int_{s \in \bR^p} (s+\lambda (\Gamma^{-1}+\lambda I)^{-1}  t)^T \Gamma^{-2} W(s+\lambda (\Gamma^{-1}+\lambda I)^{-1}  t) e^{-\frac{1}{2} \|s \|_{(\Gamma^{-1}+\lambda I)^{-1}}^2  } \id s\\
 &= \int_{s \in \bR^p} s^T \Gamma^{-2}W s e^{- \frac{1}{2}\|s\|_{(\Gamma^{-1}+\lambda I)^{-1}}^2  } \id s \\
 &+ 2  \int_{s \in \bR^p} (\lambda (\Gamma^{-1}+\lambda I)^{-1}  t)^T \Gamma^{-2} Ws e^{-\frac{1}{2} \|s \|_{(\Gamma^{-1}+\lambda I)^{-1}}^2  } \id s \\
 &+   (\lambda (\Gamma^{-1}+\lambda I)^{-1}  t)^T \Gamma^{-2} W(\lambda (\Gamma^{-1}+\lambda I)^{-1}  t)  \int_{s \in \bR^p} e^{-\frac{1}{2} \|s \|_{(\Gamma^{-1}+\lambda I)^{-1}}^2  } \id s .
\end{split}
  \end{equation}
 The second term is 0 because $s \to se^{-\|s\|_{(\Gamma^{-1}+\lambda I)^{-1}}^2}$  is odd. With $\mathbf{C}(X)$ the normalization constant of the Gaussian of covariance matrix $X$, and 
 \begin{equation}
 \mathbf{D}(X,Y)=  \int_{s \in \bR^p} s^T Xse^{-\frac{1}{2} \|s\|_{Y}^2  } \id s  ,
 \end{equation}
  we have 

\begin{equation}\label{eq:varcov2}
 \begin{split}
   B&=\mathbf{D}\big(\Gamma^{-2}W, (\Gamma^{-1}+\lambda I)^{-1} \big) +   \lambda^2  t^T \Gamma^{-2} (\Gamma^{-1}+\lambda I)^{-2} W  t \mathbf{C}{(\Gamma^{-1}+\lambda I)^{-1}}\\
   &=\mathbf{D}\big(\Gamma^{-2}W, (\Gamma^{-1}+\lambda I)^{-1} \big) +   \lambda^2  t^T (I+\lambda \Gamma)^{-2} W  t \mathbf{C}{(\Gamma^{-1}+\lambda I)^{-1}} .
   \end{split}
\end{equation}
 
 Going back to the  full integral,~\eqref{eq:varcov1} and~\eqref{eq:varcov2} yield
 
  \begin{equation}
 \begin{split}
 \|\partial_{W} \mu_{0,\Gamma}\|_K^2  =&\frac{1}{4} \int_{t \in \bR^p}\left(\mathbf{D}(\Gamma^{-2}W, (\Gamma^{-1}+\lambda I)^{-1}) + \lambda^2  t^T (I+\lambda \Gamma)^{-2} W  t \mathbf{C}{(\Gamma^{-1}+\lambda I)^{-1}}\right)\\
 &e^{-\frac{1}{2}\lambda\|t\|_2^2 + \frac{1}{2}\|\lambda (\Gamma^{-1}+\lambda I)^{-1}  t\|_{(\Gamma^{-1}+\lambda I)^{-1}}^2}t^T \Gamma^{-2} Wt e^{-\frac{1}{2}\|t\|_\Gamma^2} \id t.
\end{split}
  \end{equation}
 We have 
 
  \begin{equation}
 \begin{split}
 &-\lambda\|t\|_2^2 + \|\lambda (\Gamma^{-1}+\lambda I)^{-1}  t\|_{(\Gamma^{-1}+\lambda I)^{-1}}^2-\|t\|_\Gamma^2 \\
 &= -t^T(\Gamma^{-1}+\lambda I)t + t^T (\Gamma^{-1}+\lambda I)\lambda^2 (\Gamma^{-1}+\lambda I)^{-1}  (\Gamma^{-1}+\lambda I)^{-1}  t \\
 &= -t^T (\Gamma^{-1}+\lambda I - \lambda^2 (\Gamma^{-1}+\lambda I)^{-1})t\\
 &= -\|t\|_{ (\Gamma^{-1}+\lambda I - \lambda^2 (\Gamma^{-1}+\lambda I)^{-1})^{-1} }^2 .
\end{split}
  \end{equation}
  
Let $Z = (\Gamma^{-1}+\lambda I - \lambda^2 (\Gamma^{-1}+\lambda I)^{-1})^{-1} $, we have 
\begin{equation}\label{eq:varcov_limz}
\begin{split}
  Z_i^{-1} &= \Gamma_i^{-1}+\lambda  - \frac{\lambda^2}{\Gamma_i^{-1}+\lambda}  = \frac{ \Gamma_i^{-2}+ 2 \lambda \Gamma_i^{-1}}{\Gamma_i^{-1} +\lambda}= \frac{ (\Gamma_i^{-1}+ 2 \lambda) \Gamma_i^{-1}}{\Gamma_i^{-1} +\lambda}
. 
\end{split}
\end{equation}
Hence, we have the following equivalent when $\lambda \to+ \infty$:
\begin{equation}
 Z_i \sim \frac{ \Gamma_i}{2}.
\end{equation}

Let us set
\begin{equation}
 \mathbf{E}(X,Y,Z) :=\int s^TXs s^TYs e^{-\frac{1}{2}\|s\|_Z^2}ds .
 \end{equation}

We then have
\begin{equation}
 \begin{split}
 \|\partial_{W} \mu_{0,\Gamma}\|_K^2  =&\frac{1}{4} \mathbf{D}\big(\Gamma^{-2}W, (\Gamma^{-1}+\lambda I)^{-1} \big) \mathbf{D}\big(\Gamma^{-2}W, (\Gamma^{-1}+\lambda I - \lambda^2 (\Gamma^{-1}+\lambda I)^{-1})^{-1}\big) \\
 &+\frac{\lambda^2}{4}\mathbf{C}{(\Gamma^{-1}+\lambda I)^{-1}} \mathbf{E}\big(\Gamma^{-2}W,(I+\lambda \Gamma)^{-2} W, (\Gamma^{-1}+\lambda I - \lambda^2 (\Gamma^{-1}+\lambda I)^{-1})^{-1}\big).
\end{split}
\end{equation}

\noindent \textbf{Step 2: dependency of $\mathbf{D}$ on $W$}

We explicit the dependency of $\mathbf{D}$  on $W$: 
\begin{equation}
\begin{split}
\mathbf{D}(\Gamma^{-2}W, Y) &=\int_{s \in \bR^p} s^T \Gamma^{-2}Wse^{- \frac{1}{2}\|s\|_{Y}^2  }\id s \\
              &=\int_{s \in \bR^p} (\sum_i s_i^2 \Gamma_i^{-2}W_i )e^{-\frac{1}{2} \|s\|_{Y}^2  } \id s\\
              &=\sum_i \Gamma_i^{-2}W_i \int_{s \in \bR^p}  s_i^2  e^{- \frac{1}{2}\|s\|_{Y}^2  } \id s\\
               &=\sum_i \Gamma_i^{-2}W_i \int_{s \in \bR^p}  s_i^2  \prod_j e^{-\frac{1}{2} \frac{|s_j|^2}{Y_j}  } \id s\\
               &=\sum_i \Gamma_i^{-2}W_i  \sqrt{\pi}^{p-1}\sqrt{\prod_{j \neq i }Y_j} \int_{s_i \in \bR}  s_i^2   e^{- \frac{1}{2}\frac{|s_i|^2}{Y_i}  }\id s_i .\\
\end{split}
\end{equation}

We make the change of variable $s_i = \sqrt{Y_i} u $ and use the fact that $\int_{\bR}  e^{-\frac{1}{2}u^2} \id u = \int_{\bR} u^2 e^{-\frac{1}{2}u^2} \id u =\sqrt{2\pi}$.  We get 
\begin{equation} \label{eq67}
\begin{split}
\mathbf{D}(\Gamma^{-2}W, Y)  &=\sum_i \Gamma_i^{-2}W_i \left(\sqrt{2\pi}\right)^{p-1}\sqrt{\prod_{j \neq i }Y_j}\int_{u \in \bR} Y_i^{\frac{3}{2}}  u^2   e^{-\frac{1}{2} |u|^2 } \id u\\
&=\left(\sqrt{2\pi}\right)^{p}  \sqrt{\prod_{j }Y_j} \sum_i Y_i \Gamma_i^{-2} W_i .\\
\end{split}
\end{equation}

\noindent \textbf{Step 3: dependency of $\mathbf{E}$ on $W$}

We calculate $\mathbf{E}\big(\Gamma^{-2}W,(I+\lambda \Gamma)^{-2} W, Z\big)$ 

\begin{equation}
 \begin{split}
  \mathbf{E}\big(\Gamma^{-2}W,(I+\lambda \Gamma)^{-2} W, Z\big)&= \sum_{i,j} (1+\lambda \Gamma_i)^{-2} \Gamma_j^{-2}W_iW_j \int_{\bR} s_i^2 s_j^2 e^{-\frac{1}{2}\|s\|_Z^2} \id s .\\
  \end{split}
\end{equation}

We make the change of variable $s_i = \sqrt{Z_i} u_i $.  We get

\begin{equation}
 \begin{split}
   \sum_{i\neq j} (1+\lambda \Gamma_i)^{-2} \Gamma_j^{-2}W_iW_j &\int_{\bR} s_i^2 s_j^2 e^{-\frac{1}{2}\|s\|_Z^2} \id s \\
   &= \left(\sqrt{2\pi}\right)^{2} \left(\sqrt{2\pi}\right)^{p-2}\sqrt{\prod_{i} Z_i}  \sum_{i\neq j} (1+\lambda \Gamma_i)^{-2} Z_i  Z_j \Gamma_j^{-2}W_iW_j  \\
   &= \left(\sqrt{2\pi}\right)^{p}\sqrt{\prod_{i} Z_i}  \sum_{i\neq j} (1+\lambda \Gamma_i)^{-2} Z_i  Z_j \Gamma_j^{-2}W_iW_j  \\
  \end{split}
\end{equation}
and, using the fact that $ \int_{u \in \bR}  u^4   e^{- \frac{1}{2}|u|^2  } \id u =3\sqrt{2\pi}$, we get

\begin{equation}
 \begin{split}
   \sum_{i} (1+\lambda \Gamma_i)^{-2} \Gamma_i^{-2}W_i^2 \int_{\bR} s_i^4 e^{-\frac{1}{2}\|s\|_Z^2} &= 3\sqrt{2\pi}\sqrt{2\pi}^{p-1}  \sqrt{\prod_{i} Z_i}  \sum_{i} (1+\lambda \Gamma_i)^{-2} \Gamma_i^{-2}Z_i^{2}W_i^2 \\
   &= 3 \left(\sqrt{2\pi}\right)^{p} \sqrt{\prod_{i} Z_i}  \sum_{i} (1+\lambda \Gamma_i)^{-2} \Gamma_i^{-2}Z_i^{2}W_i^2 .
  \end{split}
\end{equation}

Hence
\begin{equation}
 \begin{split}
  \mathbf{E}\big(\Gamma^{-2}W,(I+\lambda \Gamma)^{-2} W, Z\big)=&\sqrt{\prod_{i} Z_i} \Big( 3 \left(\sqrt{2\pi}\right)^p \sum_{i} (1+\lambda \Gamma_i)^{-2} \Gamma_i^{-2}Z_i^2W_i^2 \\
  -& \left(\sqrt{2\pi}\right)^p \sum_{i} (1+\lambda \Gamma_i)^{-2}  \Gamma_i^{-2} Z_i^2 W_i^2 \\
  +& \left(\sqrt{2\pi}\right)^p(\sum_{i} (1+\lambda \Gamma_i)^{-2} Z_i W_i)(\sum_{j} Z_j \Gamma_j^{-2}W_j  ) \Big)\\
  =& \sqrt{\prod_{i} Z_i} \left(\sqrt{2\pi}\right)^p\Big( 2\sum_{i} (1+\lambda \Gamma_i)^{-2}  \Gamma_i^{-2} Z_i^2 W_i^2 \\
  +& (\sum_{i} (1+\lambda \Gamma_i)^{-2} Z_i W_i)(\sum_{j} Z_j \Gamma_j^{-2}W_j  )\Big).
  \end{split}
\end{equation}

\noindent \textbf{Step 4: computing a lower bound for $\|\partial_{W} \mu_{0,\Gamma}\|_K^2$}

Using $\|W\|_F^2 =1 $, we have
\begin{equation}
 \begin{split}
  &\mathbf{E}\big(\Gamma^{-2}W,(I+\lambda \Gamma)^{-2} W, Z\big) \\
  &\geq \left(\sqrt{2\pi}\right)^p \sqrt{\prod_{i} Z_i} \Big(  2\inf_{i} (1+\lambda \Gamma_i)^{-2}  \Gamma_i^{-2}Z_i^2
  + (\sum_{i} (1+\lambda \Gamma_i)^{-2} Z_i W_i)(\sum_{j} Z_j \Gamma_j^{-2}W_j  ) \Big).
  \end{split}
\end{equation}
Also,  we  have
\begin{equation}
  \text{sign} \Big((\sum_i  (1+\lambda \Gamma_i)^{-2} Z_i W_i)(\sum_{j} Z_j \Gamma_j^{-2}W_j  )\Big)
  =  \text{sign}\Big(\lambda^2(\sum_i  (1+\lambda \Gamma_i)^{-2} Z_i W_i)(\sum_{j} Z_j \Gamma_j^{-2}W_j  )\Big)
  \end{equation}
and, using that $Z_i^{-1} \to_{\lambda \to \infty} 2\Gamma_i^{-1} $ (from~\eqref{eq:varcov_limz}), 
\begin{equation}
\begin{split}
\lambda^2(\sum_i  (1+\lambda \Gamma_i)^{-2} Z_i W_i)(\sum_{j} Z_j \Gamma_j^{-2}W_j  )) &=  (\sum_i  \lambda^2(1+\lambda \Gamma_i)^{-2} Z_i W_i)(\sum_{j} Z_j \Gamma_j^{-2}W_j  )) \\
 &\to_{\lambda \to \infty} 
\frac{1}{4} (\sum_{j} \Gamma_j^{-1} W_j)^2>0
 .\\
\end{split}
\end{equation}

Hence, for $\lambda$ large enough, $(\sum_i  (1+\lambda \Gamma_i)^{-2} Z_i W_i)(\sum_{j} Z_j \Gamma_j^{-2}W_j  ) \geq 0 $.  This implies 
\begin{equation}
 \begin{split}
  \mathbf{E}\big(\Gamma^{-2}W,(I+\lambda \Gamma)^{-2} W, Z\big)
  \geq &  \sqrt{\prod_{i} Z_i}\left(\sqrt{2\pi}\right)^{p} 2  \inf_{i} (1+\lambda \Gamma_i)^{-2}  \Gamma_i^{-2}Z_i^2 \\
  \end{split}
\end{equation}

Moreover, writing $Y_i =(\Gamma_i^{-1}+\lambda I)^{-1} $ we have 
\begin{equation}
 \begin{split}
  \mathbf{C}{(\Gamma^{-1}+\lambda I)^{-1}} = \sqrt{2\pi}^p \sqrt{\prod_i Y_i} .
  \end{split}
\end{equation}
Putting everything together, for $\lambda$ large enough, we deduce from \eqref{eq67} that
\begin{equation}
 \begin{split}
 \|\partial_{W} \mu_{0,\Gamma}\|_K^2  \geq& \frac{1}{4}\left(\sqrt{2\pi}\right)^{2p} \sqrt{\prod_{i }Y_i} \sum_i Y_i \Gamma_i^{-2} W_i   \sqrt{\prod_{i }Z_i} \sum_i Z_i \Gamma_i^{-2} W_i \\
 &+ \frac{\lambda^2 }{4}\left(\sqrt{2\pi}\right)^{2p} \sqrt{\prod_i Y_i}  \sqrt{\prod_{i} Z_i}\cdot 2  \inf_{i} (1+\lambda \Gamma_i)^{-2}  \Gamma_i^{-2}Z_i^2 \\
 =&\frac{1}{4}\left(2\pi\right)^{p}  \sqrt{\prod_i Y_iZ_i} \big( ( \sum_i Y_i \Gamma_i^{-2} W_i  )( \sum_i Z_i \Gamma_i^{-2} W_i) +2\lambda^2    \inf_{i} (1+\lambda \Gamma_i)^{-2}  \Gamma_i^{-2} Z_i^2\big)\\
 \geq&\frac{1}{4} \left(2\pi\right)^{p} \sqrt{\prod_i Y_iZ_i} \big( O(\frac{1}{\lambda}) +2\lambda^2    \inf_{i} (1+\lambda \Gamma_i)^{-2}  \Gamma_i^{-2}Z_i^2 \big)\\
 &\sim_{\lambda \to \infty} \frac{1}{8}\left(2\pi\right)^{p} \sqrt{\prod_i Y_i Z_i}   \inf_i \Gamma_i^{-2} .
\end{split}
\end{equation}
Hence there is $\lambda$ large enough such that $\|\partial_{W} \mu_{0,\Gamma}\|_K^2 $ is lower bounded by a positive constant that depends on $\lambda$ and $\Gamma$.
\end{proof}

\section*{Acknowledgement}
Y. Traonmilin acknowledges the support of the French Agence Nationale de la Recherche (ANR) under reference ANR-20-CE40-0001 EFFIREG. J-F Aujol  acknowledges the support of the French Agence Nationale de la Recherche (ANR) under reference ANR-18-CE92-0050 SUPREMATIM.

  \bibliographystyle{imaiai}
   \bibliography{non_convex_optim}
\end{document}